\pdfoutput=1

\documentclass[12pt]{article}
\usepackage[letterpaper, margin=1.0in]{geometry}
\providecommand{\keywords}[1]{\noindent\textbf{{Keywords:}} #1.}
\newcommand{\captionDesc}[1]{\caption{#1}}
\usepackage[backend=biber, style=numeric, sorting=nyt, maxbibnames=10, maxalphanames=10, minalphanames=10]{biblatex}
\usepackage[T1]{fontenc}
\usepackage[english]{babel}
\usepackage{csquotes}
\usepackage{algorithm2e}
\RestyleAlgo{ruled}

\newcommand{\define}{\stackrel{\mathclap{\mbox{\text{\tiny def}}}}{=}}
\usepackage{mathtools}

\usepackage{amsthm}
\usepackage{thmtools, thm-restate}
\newtheorem{definition}{Definition}[section]

\newtheorem{proposition}[definition]{Proposition}

\newtheorem*{example*}{Example}

\theoremstyle{remark}

\usepackage{graphicx} \graphicspath{{./figures/}}
\usepackage{tikz}
\tikzset{font={\fontsize{10pt}{12}\selectfont}}
\usepackage{pgfplots} \pgfplotsset{compat=newest}
\usepackage{quiver}
\usepackage{placeins}
\usepackage{float}

\usepackage{booktabs,multirow}
\usepackage{booktabs}

\usepackage{url} \PassOptionsToPackage{hyphens}{url} 
\usepackage{hyperref} \PassOptionsToPackage{breaklinks}{hyperref}

\usepackage{xcolor}

\hypersetup{colorlinks, linkcolor={red!50!black}, citecolor={blue!50!black}, urlcolor={blue!80!black}}

\usepackage{enumitem}

\usepackage[capitalise]{cleveref}
\crefname{definition}{Definition}{Definitions}
\crefname{proposition}{Proposition}{Propositions}

\crefname{theorem}{Theorem}{Theorems}
\crefname{corollary}{Corollary}{Corollaries}
\crefname{example}{Example}{Examples}
\crefname{remark}{Remark}{Remarks}

\usepackage[symbols, acronym, nonumberlist, nogroupskip, stylemods={mcols,longbooktabs}, section=subsection, numberedsection]{glossaries-extra} \setabbreviationstyle[acronym]{long-short} \glssetcategoryattribute{symbol}{nohyper}{true} \glssetcategoryattribute{acronym}{nohyper}{true}

\newacronym{ADS}{ADS}{authenticated data structure}
\newacronym{MT}{MT}{Merkle tree}
\newacronym{HMT}{HMT}{Huffman-Merkle Tree}

\glsxtrnewsymbol[description={A tree.}]{tree}{\ensuremath{\tau}}

\newcommand{\dsname}{\gls{HMT}\xspace}

\glsxtrnewsymbol[description={Leaf number.}]{leaf}{\ensuremath{\ell}}

\begin{document}
\title{Authenticated Data Structures for Dynamic Workloads}
\author{Ziheng (Tom) Shangguan$^1$, Aviv Yaish$^{2,3,4}$, Dahlia Malkhi$^5$}
\date{\small{$^1$Brown University, $^2$Yale University, $^3$IC3, $^4$Complexity Science Hub Vienna, $^5$UCSB}}
\maketitle

\begin{abstract}
We introduce the \emph{Huffman-Merkle Tree (HMT)}, an authenticated data structure (ADS) for dynamic workloads where items may differ in access frequencies, and access frequencies can change over time.
An ADS allows proving item membership against a short commitment to a large mutable state, with applications including verifiable storage, Internet transparency services, and blockchains.
Optimizing ADS performance under continuously changing access frequencies has not been fully addressed before, neither in theory nor in practice.

HMT addresses dynamically changing access skew through two complementary mechanisms.
The first is a Huffman-coding-based Merkle-tree layout, with a novel extension to support evolving access frequencies.
The second is an elastic tiering regime that partitions items across separate trees, such as hot and cold tiers, with adaptive migration between them.
The key insight in the combined approach is to place frequently accessed items closer to the root, while assigning less frequently accessed items to progressively larger and deeper trees.
This reduces the overall frequency-weighted access cost.
Additionally, our scheme is designed to scale to gigabytes of data spanning millions of items.
To handle dynamism efficiently, layout updates are applied in batches, access frequencies are tracked using a count-min sketch, and the system employs a tier-promotion cache while exploring multiple tier-migration policies.

We implement HMT and compare it on real-world data with Ethereum's Merkle Patricia Trie (MPT) ADS and its proposed replacement, the Unified Binary Tree (UBT).
Our evaluation considers two metrics: the amount of hashing per ADS update and access-weighted membership-proof size.
The latter metric captures both each item's access cost and frequency.
We find that the best HMT policy uses about $2.4\times$ and $0.34\times$ less average hash operations than MPT and UBT respectively, together with $0.18\times$ and $0.55\times$ shorter proofs.

\end{abstract}

\section{Introduction}
Authenticated Data Structures (ADSs) allow committing to a large state using a short digest and proving against the digest that specific elements are included.
They are useful for memory checking \cite{wang2023locality}, remote storage \cite{papamanthou2008authenticated,miller2014authenticated}, internet certificates \cite{bonneau2025merkle}, and blockchains \cite{mendonca2026efficient}.

The Merkle Tree (MT) is a widely used ADS where elements appear as leaves, each internal node stores a hash of the concatenation of its children, and the root hash serves as a compact commitment to the entire structure.
Item inclusion is proven by showing the sibling hashes along the path from the element to the root.
So, inclusion proofs in a balanced binary MT with $N$ leaves contain $log N$ hash values, incurring logarithmic cost regardless of element access frequencies.
This shortcoming extends to other notable ADSs, e.g., Merkle-Patricia Tries (MPTs)~\cite{wood2025ethereum}, which also suffer costs that are independent of access frequencies.

Real-world workloads are typically skewed, with some elements accessed more frequently than others~\cite{guo2019graph}.
Our focus is therefore the \emph{average-case} performance of ADSs under such workloads, where each element's proof and update cost is weighted by its access frequency.

When all access frequencies are known in advance, this optimization problem is well understood: the optimal static tree can be described by a Huffman code \cite{moffat2019huffman}.
However, the dynamic setting is more challenging.
Prior work has addressed changing access distributions by rebuilding the entire data structure whenever the distribution changes~\cite{mizrahi2024traffic}.
However, this incurs overhead which must be accounted for in the same concrete units as proof generation and verification: the number of hash-compression invocations.
To our knowledge, the practical challenge of devising an ADS which efficiently adapts to evolving access frequencies while explicitly accounting for restructuring cost has not been addressed before.

Given these gaps, we ask:
\begin{quote}
    \emph{Can we design efficient authenticated data structures for applications where elements may be updated and access frequencies can change?}
\end{quote}

\begin{figure*}
    \centering
    \begin{center}
    \begin{tabular}{ccc}
    \scalebox{0.7}{
    \begin{tikzpicture}[commutative diagrams/every diagram]
        \node (H1234) at (3.75,3) {$H(H_{1,2}||H_{3,4})$};
        \node (H12)   at (1.25,2) {$H_{1,2} \coloneqq H(H_1||H_2)$};
        \node (H34)   at (6.25,2) {$H_{3,4} \coloneqq H(H_3||H_4)$};
        \node (H1)    at (0,   1) {$H_1 \coloneqq H(m_1)$};
        \node (H2)    at (2.5, 1) {$H_2 \coloneqq H(m_2)$};
        \node (H3)    at (5,   1) {$H_3 \coloneqq H(m_3)$};
        \node (H4)    at (7.5, 1) {$H_4 \coloneqq H(m_4)$};
        \node (m1)    at (0,   0) {$m_1$};
        \node (m2)    at (2.5, 0) {$m_2$};
        \node (m3)    at (5,   0) {$m_3$};
        \node (m4)    at (7.5, 0) {$m_4$};
        \path[commutative diagrams/.cd, every arrow, every label]
        (H12) edge node {} (H1234)
        (H34) edge node {} (H1234)
        (H1)  edge node {} (H12)
        (H2)  edge node {} (H12)
        (H3)  edge node {} (H34)
        (H4)  edge node {} (H34)
        (m1)  edge node {} (H1)
        (m2)  edge node {} (H2)
        (m3)  edge node {} (H3)
        (m4)  edge node {} (H4);
    \end{tikzpicture}
    }
    &&
    \scalebox{0.7}{
    \begin{tikzpicture}[commutative diagrams/every diagram]
        \node (H1234) at (0,3) {$H(H_1||H_{2,3,4})$};
        \node (H1)    at (-1,  2) {$H_1 \coloneqq H(m_1)$};
        \node (m1)    at (-1,  1) {$m_1$};
        \node (H234)  at (2,   2) {$H_{2,3,4} \coloneqq H(H_2||H_{3,4})$};
        \node (H2)    at (1,   1) {$H_2 \coloneqq H(m_2)$};
        \node (m2)    at (1,   0) {$m_2$};
        \node (H34)   at (4,   1) {$H_{3,4} \coloneqq H(H_3||H_4)$};
        \node (H3)    at (2.8, 0) {$H_3 \coloneqq H(m_3)$};
        \node (H4)    at (5.2, 0) {$H_4 \coloneqq H(m_4)$};
        \node (m3)    at (2.8,-1) {$m_3$};
        \node (m4)    at (5.2,-1) {$m_4$};
        \path[commutative diagrams/.cd, every arrow, every label]
        (H1)  edge node {} (H1234)
        (m1)  edge node {} (H1)
        (H234) edge node {} (H1234)
        (H2)  edge node {} (H234)
        (H34)  edge node {} (H234)
        (m2)  edge node {} (H2)
        (H3)  edge node {} (H34)
        (H4)  edge node {} (H34)
        (m3)  edge node {} (H3)
        (m4)  edge node {} (H4);
    \end{tikzpicture}
    }
    \\
    (a) && (b)
    \end{tabular}
    \end{center}
    \captionDesc{Panel (a) depicts a Merkle Tree (MT) and (b) shows a Huffman Merkle Hash Tree (HuffMHT), both with 4 items with monotonically decreasing access frequencies, i.e., $m_i$ is accessed more frequently than $m_{i+1}$.}
    \label{fig:HMT}
\end{figure*}

\subsection{This Work}
We introduce the \dsname, an efficient ADS for dynamic workloads.
HMT separates state into tiers:
frequently-accessed ``hot'' items reside in a small Huffman Merkle Hash Tree (HuffMHT) and a balanced Merkle Tree (MT) stores cold elements.
Our tiered approach is designed to obtain the ``best of both worlds'': HuffMHT obtains good average runtime complexity for operations involving hot items, while the cold tier is flatter and thus more suitable for efficient insertions of new elements given there are many rarely-accessed items.
We furthermore optimize our ADS for dynamic workloads by batching and incrementally updating its layout when access frequencies change.

Our approach builds on two observations.
First, authentication can be separated from layout optimization, that is, value updates are reflected in the next committed root of an item's tier, while layout changes are batched.
Items admitted to the hot tier between rebuilds are authenticated through an overflow MT, so HMT never needs to restructure the hot HuffMHT after every access.
Second, given a monotone initial frequency distribution, a single access requires only swapping two nodes and updating hashes along the affected paths.
The cost of a single swap is proportional to the distance between the relevant pair of nodes, and several swaps can be batched to avoid redundant hashing and thus lower update overhead.

HMT uses lightweight, deterministic metadata to decide which elements should be considered as hot.
A Count-Min Sketch (CMS) efficiently estimates access frequencies without maintaining exact counters for every observed key \cite{cormode2005countmin}.
A bucketed least-frequently-used (LFU) cache tracks candidate elements for promotion and demotion.
On top of this, HMT supports several migration policies:
(1) Ratio-Based uses lifetime access rate and is simple and stable.
(2) Sliding-Window uses recent accesses and is more responsive to workload shifts.
(3) Dynamic-Control adjusts admission thresholds and tier capacity using feedback from occupancy and rejected promotions.
Policy decisions are batched; in our blockchain evaluation, this takes place at the end of the block-building process.

Using real-world data, we compare five ADSs (MPT, UBT, and three HMT configurations) on two metrics.
First, commit-time hash input measures the hashed bytes amount, thus capturing the work induced by updates while avoiding implementation- and hardware-specific differences.
Second, \emph{average} membership-proof length, weighted by the frequency with which each item is accessed.

Our evaluation shows that HMT improves on prominent ADSs, and that Sliding-Window is the best-performing policy: it uses about $2.4\times$ less average hash operations than MPT and about $0.34\times$ less than UBT.
Its access-weighted proofs are also substantially smaller, reducing the plotted proof sizes by $0.18\times$ than MPT and $0.55\times$ than UBT.
Ratio-Based HMT is slightly less responsive but remains stable, while Dynamic-Control HMT demonstrates an adaptive admission mechanism that remains competitive in proof size at the cost of additional hash-input work in this run.
These results indicate that frequency-aware authenticated layouts can improve both witness size and update overhead on real workloads.

\paragraph{Contributions}
To summarize our contributions:
\begin{enumerate}
    \item We propose HMT, a frequency-aware ADS which exploits skewed access distributions to tailor performance to real-world workloads while supporting updates, insertions, deletions, deterministic root computation, and bounded restructuring work.
    HMT combines a cold authenticated tier, a periodically rebuilt HuffMHT hot tier, and an overflow tree for newly hot elements admitted between rebuilds.

    \item We develop efficient mechanisms for adapting HMT under dynamic workloads, including Count-Min Sketch (CMS) frequency estimation, a bucketed least-frequently-used (LFU) promotion cache, lifetime-rate promotion, sliding-window promotion with lazy demotion, and feedback-controlled admission.

    \item We implement HMT and evaluate it on real-world Ethereum data, comparing its performance against Ethereum's current MPT design and its proposed UBT replacement.
    The evaluation shows HMT obtains improved performance with respect to both hashing overhead and access-weighted membership-proof size.
\end{enumerate}

\paragraph{Paper Structure}
In \cref{sec:preliminaries}, we cover preliminaries required for our work.
Next, we describe different approaches for efficiently optimizing ADS layouts when item access frequencies can dynamically change in \cref{sec:update-strategies}.
We proceed by specifying the tiered architecture of our \dsname in \cref{sec:hmt-arch}, and its dynamic tier adjustment policies in \cref{sec:dynamic}.
Then, we evaluate several HMT configurations against prominent ADSs in \cref{sec:evaluation}.
Finally, we go over related work in \cref{sec:related} and conclude in \cref{sec:conclusion}.
In the interest of space, some algorithmic details are relegated to \cref{sec:algorithms} and part of the evaluation to \cref{sec:additional-policy-replay}.

\section{Preliminaries}
\label{sec:preliminaries}
We now cover the foundations of ADSs, their application to blockchain state management, and the principles of coding theory utilized in our proposed architecture.

\subsection{Authenticated Data Structures} \label{ADS_Intro}
An ADS allows a prover to convince a verifier that an element is a member of a dataset represented by a succinct commitment, without requiring the verifier to see the entire dataset. 
ADSs expose an API for reading data, generating membership proofs, updating data, committing to the current state, and verifying membership proofs against the committed state, as formalized in \cref{def:ads-api}.

\begin{definition}[Authenticated Data Structure Interface]
    \label[definition]{def:ads-api}
An Authenticated Data Structure (ADS) over a dataset $D$ of key-value pairs $(k,v)$ exposes the following methods:
\begin{itemize}
    \item $\mathsf{Put}(D,k,v) \rightarrow D'$: inserts or updates initial state $D$ with the key-value pair $(k,v)$, producing an updated state $D'$ such that $(k,v) \in D'$.
    \item $\mathsf{Delete}(D,k) \rightarrow D'$: removes a key-value pair $(k,v)$ from initial state $D$, returning an updated state $D' \leftarrow D \setminus \{(k,v)\}$.
    \item $\mathsf{Commit}(D) \rightarrow c$: returns a commitment $c$ of $D$.
    \item $\mathsf{Get}(D,k) \rightarrow \{v, \bot\}$: returns the value $v$ associated with the key $k$ in state $D$ if it is indeed present.
    Otherwise, $k$ is absent and $\bot$ is returned.
    \item $\mathsf{GetProof}(D,k) \rightarrow \{\pi, \bot\}$: returns a proof $\pi(k, v, c)$ against $c = \mathsf{Commit}(D)$ if key $k$ is present in state $D$.
    Otherwise, returns $\bot$.
    \item $\mathsf{Verify}(c,k, v,\pi) \rightarrow \{0,1\}$: checks whether a proof $\pi$ authenticates key-value pair $(k, v)$ against a commitment $c$, where for $\pi = \mathsf{GetProof}(D,k) $ and given $c = \mathsf{Commit}(D)$, then $Verify(c, k, v, \pi)$ returns 1 iff $(k,v) \in D$.
\end{itemize}
\end{definition}

\subsection{Merkle Trees (MTs)}
A Merkle Tree (MT) is a cryptographic accumulator where elements reside in leaf nodes~\cite{merkle1987digital}.
Each internal node stores the cryptographic hash of the concatenation of its children's hashes, and the tree's root serves as a commitment to the entire dataset. 
We describe MT inclusion-proof generation and verification in \Cref{alg:mt-proof}.
Briefly, an inclusion proof for a leaf $x$ consists of the sibling hash and left/right orientation at each level on the path from $x$ to the root.
Thus, proof size is logarithmic in the number of leaves $N$: $O(\log_2 N)$.
Given a proof, a verifier hashes $x$'s value, iteratively combines the running hash with each sibling in the indicated order, and accepts if the final hash equals the trusted root commitment.

\subsection{Blockchain State}
Ethereum processes state-modifying transactions in batches called blocks, where each blocks contains a commitment to the blockchain's state up to that point~\cite{wood2025ethereum}.
The Merkle Patricia Trie (MPT) ADS is currently used to maintain this commitment, with the Unified Binary Tree (UBT) proposed as a future replacement \cite{buterin2025eip7864}.
We briefly describe both here for the sake of completeness, as we compare our solution with these structures in \cref{sec:evaluation}.

\paragraph{Merkle Patricia Trie (MPT)}
The MPT supports membership and absence proofs, but its witnesses can be large because proofs contain encoded trie nodes, rather than the standard MT's approach of having just one sibling hash per level.
In Ethereum, a tree-of-tries layout is used for account and contract storage~\cite{wood2025ethereum}, with the choice of MPT specifically has been found to make Ethereum susceptible to Denial-of-Service (DoS) attacks \cite{yaish2024speculative}.
At a high level, Ethereum's MPT proof interface converts the queried key into a hexary so-called ``nibble'' path and returns the encoded trie nodes encountered along that path.
\Cref{alg:mpt-proof} illustrates this inclusion-proof procedure, and \cref{fig:mpt} shows an example MPT.

\begin{figure*}
\centering
\begin{minipage}[t]{0.48\textwidth}
\centering
\resizebox{\linewidth}{!}{%
\begin{tikzpicture}[
  x=0.82cm,y=0.82cm,
  >=latex,
  font=\small,
  box/.style={
    draw,
    align=center,
    inner sep=2.5pt,
    minimum height=7mm
  },
  pathlbl/.style={font=\scriptsize, fill=white, inner sep=1pt}
]
\draw (-1.4,0) rectangle (1.4,1);
\draw (-1.4,0.55) -- (1.4,0.55);
\node at (0,0.75) {Root};
\node at (0,0.24) {$H(H_1,H_{23})$};
\draw (-4.4,-2.0) rectangle (4.4,-0.7);
\draw (-4.4,-1.3) -- (4.4,-1.3);

\draw (-3.3,-2.0) -- (-3.3,-1.3);
\draw (-2.2,-2.0) -- (-2.2,-1.3);
\draw (-1.1,-2.0) -- (-1.1,-1.3);
\draw (0.0,-2.0) -- (0.0,-1.3);
\draw (1.1,-2.0) -- (1.1,-1.3);
\draw (2.2,-2.0) -- (2.2,-1.3);
\draw (3.3,-2.0) -- (3.3,-1.3);

\node at (0,-1.02) {Branch Node};
\node at (-3.85,-1.64) {\texttt{0}};
\node at (-2.75,-1.64) {\texttt{1}};
\node at (-1.65,-1.64) {\dots};
\node at (-0.55,-1.64) {\dots};
\node at (0.55,-1.64) {\dots};
\node at (1.65,-1.64) {\texttt{e}};
\node at (2.75,-1.64) {\texttt{f}};
\draw[-latex, thick] (0,-0.7) -- (0,0);
\draw (-5.3,-4.6) rectangle (-2.4,-3.5);
\draw (-5.3,-4.05) --       (-2.4,-4.05);
\node at (-3.85,-3.78) {Leaf Node};
\node at (-3.85,-4.33) {$H_1 := H(m_1)$};

\draw[-latex, thick] (-3.85,-5.25) -- (-3.85,-4.6);
\node[pathlbl] at (-3.85,-5.47) {$m_1$};

\draw[-latex, thick]
  (-3.85,-3.5) .. controls (-3.85,-2.75) and (-3.55,-2.45) .. (-2.75,-2.0);
\node[pathlbl] at (-3.55,-2.9) {\texttt{0x1c2e}};
\draw (2.,-4.6) rectangle (6,-3.5);
\draw (2.,-4.05) --       (6,-4.05);
\node at (4.05,-3.78) {Extension Node};
\node at (4.05,-4.33) {$H_{23} := H(H_2,H_3)$};

\draw[-latex, thick]
  (4.05,-3.5) .. controls (4.05,-2.75) and (3.55,-2.45) .. (2.75,-2.0);
\node[pathlbl] at (3.75,-2.9) {\texttt{0xf}};
\draw (0.6,-6.7) rectangle (7.2,-5.4);
\draw (0.6,-6.0) -- (7.2,-6.0);

\draw (1.7,-6.7) -- (1.7,-6.0);
\draw (2.8,-6.7) -- (2.8,-6.0);
\draw (3.9,-6.7) -- (3.9,-6.0);
\draw (5.0,-6.7) -- (5.0,-6.0);
\draw (6.1,-6.7) -- (6.1,-6.0);

\node at (3.9,-5.72) {Branch Node};
\node at (1.15,-6.34) {\dots};
\node at (2.25,-6.34) {\texttt{3}};
\node at (3.35,-6.34) {\dots};
\node at (4.45,-6.34) {\dots};
\node at (5.55,-6.34) {\texttt{a}};
\node at (6.65,-6.34) {\dots};

\draw[-latex, thick] (4.05,-5.4) -- (4.05,-4.6);
\node[pathlbl] at (4.10,-5.00) {\texttt{0xb}};
\draw (-0.2,-9.4) rectangle (2.7,-8.3);
\draw (-0.2,-8.85) --       (2.7,-8.85);
\node at (1.3,-8.58) {Leaf Node};
\node at (1.3,-9.13) {$H_2 := H(m_2)$};

\draw[-latex, thick] (1.3,-10.05) -- (1.3,-9.4);
\node[pathlbl] at (1.3,-10.27) {$m_2$};

\draw[-latex, thick]
  (1.3,-8.3) .. controls (1.3,-7.55) and (1.7,-7.2) .. (2.25,-6.7);
\node[pathlbl] at (1.70,-7.45) {\texttt{0x35}};
\draw (5.2,-9.4) rectangle (8.2,-8.3);
\draw (5.2,-8.85) --       (8.2,-8.85);
\node at (6.7,-8.58) {Leaf Node};
\node at (6.7,-9.13) {$H_3 := H(m_3)$};

\draw[-latex, thick] (6.7,-10.05) -- (6.7,-9.4);
\node[pathlbl] at (6.7,-10.27) {$m_3$};

\draw[-latex, thick]
  (6.7,-8.3) .. controls (6.7,-7.55) and (6.35,-7.2) .. (5.55,-6.7);
\node[pathlbl] at (6.45,-7.45) {\texttt{0xa9}};

\end{tikzpicture}%
}
\captionDesc{An example of a Merkle Patricia Trie (MPT).}
\label{fig:mpt}
\end{minipage}\hfill
\begin{minipage}[t]{0.48\textwidth}
\centering
\resizebox{\linewidth}{!}{%
\begin{tikzpicture}[
  x=0.82cm,y=1.24cm,
  >=latex,
  font=\small,
  ann/.style={font=\footnotesize},
  pathlbl/.style={font=\scriptsize, fill=white, inner sep=1pt}
]
\draw (-1.55,0) rectangle (1.55,1.02);
\draw (-1.55,0.50) -- (1.55,0.50);
\node at (0,0.77) {Root};
\node at (0,0.23) {$H(H_a,H_b)$};
\draw (-4.9,-1.55) rectangle (-1.3,-0.48);
\draw (-4.9,-1.02) --        (-1.3,-1.02);
\node at (-3.08,-0.75) {Internal Node};
\node at (-3.08,-1.28) {$H_a := H(H_c,H_d)$};

\draw (1.28,-1.55) rectangle (4.88,-0.48);
\draw (1.28,-1.02) --       (4.88,-1.02);
\node at (3.08,-0.75) {Internal Node};
\node at (3.08,-1.28) {$H_b := \cdots$};

\draw[-latex, thick] (-3.08,-0.48) -- (0,0);
\draw[-latex, thick] (3.08,-0.48) -- (0,0);
\node[pathlbl] at (-1.55,-0.18) {\texttt{0}};
\node[pathlbl] at (1.55,-0.18) {\texttt{1}};
\draw (-5.95,-3.35) rectangle (-2.25,-2.28);
\draw (-5.95,-2.82) -- (-2.25,-2.82);
\node at (-4.10,-2.55) {Stem Node};
\node at (-4.10,-3.08) {\texttt{$H_{c} = H(H_{e},H_{f})$}};

\node[ann] at (-2.06,-2.82) {\Large $\cdots$};

\draw[-latex, thick] (-4.10,-2.28) -- (-3.08,-1.55);
\draw[-latex, thick] (-2.06,-2.28) -- (-3.08,-1.55);
\node[pathlbl] at (-3.75,-1.95) {\texttt{0}};
\node[pathlbl] at (-2.45,-1.95) {\texttt{1}};
\node[ann] at        (3.08,-2.82) {\Large $\cdots$};
\draw[-latex, thick] (3.08,-2.65) -- (3.08,-1.55);
\draw (-10.75,-4.88) rectangle (-6.25,-4.00);
\node at (-8.50,-4.44) {$H_e:=H(H_0,H_1)$};

\draw (-1.95,-4.88) rectangle (2.55,-4.00);
\node at (0.30,-4.44) {$H_f:=H(H_{254},H_{255})$};

\node[ann] at (-4.10,-4.65) {\Large $\cdots$};

\draw[-latex, thick] (-8.50,-4.00) -- (-4.10,-3.35);
\draw[-latex, thick] (0.30,-4.00) -- (-4.10,-3.35);
\draw (-12.10,-7.55) rectangle (-8.90,-6.45);
\draw (-12.10,-7.00) -- (-8.90,-7.00);
\node at (-10.50,-6.72) {Value Node};
\node[font=\footnotesize] at (-10.50,-7.27) {$H_0:=H(m_0)$};

\draw (-8.10,-7.55) rectangle (-4.90,-6.45);
\draw (-8.10,-7.00) -- (-4.90,-7.00);
\node at (-6.50,-6.72) {Value Node};
\node[font=\footnotesize] at (-6.50,-7.27) {$H_1:=H(m_1)$};

\draw (-3.30,-7.55) rectangle (-0.10,-6.45);
\draw (-3.30,-7.00) -- (-0.10,-7.00);
\node at (-1.70,-6.72) {Value Node};
\node[font=\footnotesize] at (-1.70,-7.27) {$H_{254}:=H(m_{254})$};

\draw (0.70,-7.55) rectangle (3.90,-6.45);
\draw (0.70,-7.00) -- (3.90,-7.00);
\node at (2.30,-6.72) {Value Node};
\node[font=\footnotesize] at (2.30,-7.27) {$H_{255}:=H(m_{255})$};

\node[ann] at (-4.10,-7.05) {\Large $\cdots$};

\draw[-latex, thick] (-10.50,-6.45) -- (-8.50,-4.88);
\draw[-latex, thick] (-6.50,-6.45) -- (-8.50,-4.88);
\draw[-latex, thick] (-1.70,-6.45) -- (0.30,-4.88);
\draw[-latex, thick] (2.30,-6.45) -- (0.30,-4.88);

\node at (-10.50,-8.35) {$m_0$};
\node at (-6.50,-8.35) {$m_1$};
\node at (-1.70,-8.35) {$m_{254}$};
\node at (2.30,-8.35) {$m_{255}$};
\draw[-latex, thick] (-10.50,-8.13) -- (-10.50,-7.55);
\draw[-latex, thick] (-6.50,-8.13) -- (-6.50,-7.55);
\draw[-latex, thick] (-1.70,-8.13) -- (-1.70,-7.55);
\draw[-latex, thick] (2.30,-8.13) -- (2.30,-7.55);

\end{tikzpicture}%
}
\captionDesc{An example of Unified Binary Tree (UBT).}
\label{fig:ubt}
\end{minipage}
\end{figure*}

\paragraph{Unified Binary Tree (UBT)}
EIP-7864 is an Ethereum Improvement Proposal to replace the current MPT-based state ADS with one based on a Unified Binary Tree (UBT)~\cite{buterin2025eip7864}.
It merges account data, storage, and code into one 32-byte key-value space, groups 256 related values by stem, and uses binary internal nodes to improve ordinary Merkle-witness size and proving-system compatibility~\cite{buterin2025eip7864}.
A UBT proof follows the binary path determined by the first 248 bits of the key's 31-byte stem and then authenticates the selected value inside the 256-value stem subtree, rather than traversing Ethereum's current account/storage tree-of-tries.
\Cref{fig:ubt} shows the corresponding stem-and-subtree structure used by UBT.

\subsection{Frequency-Based Structures} \label{sec:frequency-based-structures}
\paragraph{Huffman Coding}
Huffman coding is a foundational entropy encoding algorithm used for lossless data compression~\cite{huffman1952method,moffat2019huffman}.
Given a distribution $P$ over symbols, it constructs a variable-length prefix code represented by a binary tree in which higher-probability symbols receive shorter root-to-leaf paths.
Huffman coding minimizes the expected prefix-code length $\sum_x p(x)L(x)$.
Note that individual codeword lengths need not equal $\log_2 \frac{1}{p(x)}$.

\paragraph{Huffman Merkle Hash Tree (HuffMHT)}
The HuffMHT uses a frequency-aware layout to reduce access-weighted authentication path length for static access distributions~\cite{munoz2005huffmht}.
It starts from a Huffman tree over element access weights, stores elements at the leaves, stores child-hash commitments at internal nodes, and uses the root hash as the commitment to the set of values.
The authentication interface is the same as in a binary Merkle tree, where an inclusion proof for element $x$ consists of the sibling hashes along the root-to-leaf path, and the verifier recomputes the root from the claimed value and the proof.
The proof depth for $x$ is the Huffman code length $L(x)$, so the same expected-length objective $\sum_x p(x)L(x)$ becomes the expected proof depth.

\Cref{fig:HMT} contrasts a balanced MT with a Huffman layout for four elements whose access frequencies decrease from $m_1$ to $m_4$.
Real-world data often exhibits Zipfian skew, e.g., word frequencies in natural text and city population sizes~\cite{newman2005power}.
\Cref{fig:hmt-mt-proof-ratio} reports the same effect across larger skewed workloads, with the $i$-th element accessed with probability proportional to $(i+b)^{-a}$, where $b=1$ and $a$ varies, where larger values of $a$ correspond to more skewed access distributions.
However, the basic HuffMHT layout is static, so the novelty of HMT is not the use of a Huffman-shaped authenticated tree itself, but the dynamic authenticated-map design that maintains such a layout under insertions, updates, and changing access frequencies.
\begin{figure*}
    \centering
    \begin{minipage}[t]{0.48\textwidth}
        \centering
        \resizebox{\linewidth}{!}{%
        \begin{tikzpicture}
            \begin{axis}[
                scaled ticks=true,
                scaled y ticks=base 10:1,
                xmode=log,
                xmin=4,
                xmax=500000,
                xlabel={Elements},
                ylabel={Weighted Average Proof Ratio},
                xmajorgrids,
                ymajorgrids,
                ymin=0,
                ymax=1,
                legend cell align={left},
                legend style={draw=none, fill=white, fill opacity=0.9, text opacity=1, font=\small, at={(0.97,0.97)}, anchor=north east},
            ]
                \addplot [draw=red, dashed, thick, mark=none] table [col sep=comma,x=n,y=zipf_a_1_5] {figures/hmt-mt-proof-ratio.csv};
                \addlegendentry{$a=1.5$}
                \addplot [draw=teal!70!black, dashdotted, thick, mark=none] table [col sep=comma,x=n,y=zipf_a_2_0] {figures/hmt-mt-proof-ratio.csv};
                \addlegendentry{$a=2$}
                \addplot [draw=orange!80!black, densely dashed, thick, mark=none] table [col sep=comma,x=n,y=zipf_a_2_5] {figures/hmt-mt-proof-ratio.csv};
                \addlegendentry{$a=2.5$}
                \addplot [draw=purple, dotted, thick, mark=none] table [col sep=comma,x=n,y=zipf_a_3_0] {figures/hmt-mt-proof-ratio.csv};
                \addlegendentry{$a=3$}
                \addplot [draw=brown!70!black, loosely dashed, thick, mark=none] table [col sep=comma,x=n,y=zipf_a_3_5] {figures/hmt-mt-proof-ratio.csv};
                \addlegendentry{$a=3.5$}
                \addplot [draw=black, densely dashdotted, thick, mark=none] table [col sep=comma,x=n,y=zipf_a_4_0] {figures/hmt-mt-proof-ratio.csv};
                \addlegendentry{$a=4$}
            \end{axis}
        \end{tikzpicture}%
        }
        \captionDesc{Weighted average authentication-path ratio of HuffMHTs to balanced binary MTs as the number of elements increases under Zipf-distributed accesses, with $p_i \propto (i+b)^{-a}$, $b=1$, and varying $a$.}
        \label{fig:hmt-mt-proof-ratio}
    \end{minipage}\hfill
    \begin{minipage}[t]{0.48\textwidth}
        \centering
        \resizebox{\linewidth}{!}{%
        \begin{tikzpicture}
            \begin{axis}[
                xmode=log,
                xlabel={Elements},
                ylabel={Worst-Case Depth Ratio (HuffMHT/MT)},
                xmajorgrids,
                ymajorgrids,
                domain=4:1000000,
                samples=100,
            ]
            \addplot [draw=blue] {x/(ln(x)/ln(2))};
            \end{axis}
        \end{tikzpicture}%
        }
        \captionDesc{Worst-case depth ratio for a cold element in a HuffMHT vs. a balanced binary MT as the number of elements increases, illustrating the possible long-tail cost for low-frequency elements in HuffMHT.}
        \label{fig:hmt-mt-insert-ratio}
    \end{minipage}
\end{figure*}

\begin{restatable}[Zipf Distribution Weighted Average Tree Depth Ratios]{proposition}{resZipfAverageDepthRatio}
\label{res:zipf-average-depth-ratio}
Fix constants $a>1$ and $b\ge0$.
For $N\ge2$, let $P_N$ be the truncated Zipf distribution over $N$ elements, where
$
p_i=\frac{(i+b)^{-a}}{Z_N}$
and
$
Z_N=\sum_{j=1}^{N}(j+b)^{-a}.
$
Let $\bar d_{\mathrm{Huff}}(N)$ be the access-weighted average leaf depth of a HuffMHT built for $P_N$.
Let $q\ge2$ be any fixed arity and let $\bar d_{\mathrm{MT_q}}(N)=\lceil\log_q N\rceil$ be the depth of a balanced $q$-ary MT.
Then
$
\frac{\bar d_{\mathrm{Huff}}(N)}{\bar d_{\mathrm{MT_q}}(N)}
=O\left(\frac{1}{\log N}\right).
$
Moreover, for any MPT baseline whose access-weighted proof depth $\bar d_{\mathrm{MPT}}(N)$ satisfies $\bar d_{\mathrm{MPT}}(N)=\Omega(\log N)$ over the considered key-space family,
$
\frac{\bar d_{\mathrm{Huff}}(N)}{\bar d_{\mathrm{MPT}}(N)}
=O\left(\frac{1}{\log N}\right).
$
\end{restatable}

Intuitively, when $a>1$, the Zipf tail is light enough that adding more elements contributes very little additional access probability.
The entropy of the truncated access distribution therefore remains bounded as $N$ grows.
A HuffMHT can exploit this skew by assigning short paths to the frequent elements, so its access-weighted average depth stays bounded up to a constant, while a balanced binary MT gives every element a path of length $\Theta(\log N)$.
Because an inclusion proof carries one sibling hash per level, the same separation applies to access-weighted proof size.
The full proof is included in the appendix; see \Cref{sec:proofs}.
\section{Handling Dynamic Access Frequencies}
\label{sec:update-strategies}
We tailor our solution to efficiently handle dynamic workloads where new items may be continually added and access frequencies can change over time.
This contrasts with the standard Huffman coding approach, which considers a fixed frequency distribution over a known element set.
We now go over several different methods for the dynamic setting, starting with completely rebuilding the structure whenever access frequencies change a la \cite{mizrahi2024traffic}, and continuing with two novel which we use in \dsname.

\paragraph{Complete Rebuild Strategy} 

One update approach is to recompute the entire tree each time, whether for single changes or batches of them.
The cost each time is linear in the number of elements (and not in the delta that changed). 
\paragraph{Incremental Update Strategy} \label{sec:incremental-update}
If an item subset $S$ is accessed in some period, then only the ranking of items in $S$ in the access-frequency sequence ordering changes. 
Denote frequencies before and after the update respectively by $f(\cdot)$ and $f'(\cdot)$.
Let elements be sorted in descending access frequencies $e_1, e_2, \dots, e_n$.
For brevity and without loss of generality, suppose only $e_{2}$ changed its relative access frequency such that $e_{j-1} < e_2 < e_j$.
That is:
\begin{align*}
& f(e_{1})  < \mathbf{f(e_{2})}  < \cdots < f(e_{j-1}) < f(e_{j}) < \cdots < f(e_{n}); \\
& f'(e_{1}) < \cdots < f'(e_{j-1}) < \mathbf{f'(e_{2})} < f'(e_{j}) < \cdots < f'(e_{n}).
\end{align*}
Swapping $e_2$ with $e_j$ fixes the code.
This entails moving $e_3$ to $e_2$'s place, $e_4$ to $e_3$, and so on, up to $e_2$ in $e_{j-1}$'s place.
In total, $|S|$ pair-swaps transform $f$ to $f'$.
For each $e_i \in S$ whose access frequency increased from $f(e_i)$ to $f'(e_i)$, find a pair-swap element $e_j$ for $e_i$ where:
$$
e_j \triangleq \arg \min_{e} \{ f(e) \mid f(e) > f(e_i), f'(e) < f'(e_i) \}
$$

That is, $e_j$ is the item whose access frequency is the minimal such that $e_i$ and $e_j$ swap places in access-frequency ranking order.  
The tree can be updated from $f$ to $f'$ by performing these $(e_i, e_j)$ pair-swaps and updating hashes along the affected paths.
The cost of a swap is proportional to the number of affected ancestor hashes, i.e., the size of the union of the two corresponding root paths.
Reflecting all these changes in the tree can be accomplished with $|S|$ pair-swaps.  
\Cref{fig:huffmht_incremental_swap} illustrates this local update for a single pair-swap; where an element's frequency crosses another's, the two elements exchange positions and only the highlighted ancestor hashes are recomputed.

\begin{proposition}[Incremental Pair-swap Update Cost]
\label{res:incremental-pair-swap-cost}
Let $\tau$ be a HuffMHT with $n$ leaves and height $h$ before an update period.
Suppose the frequency-order change induced by an accessed set $S$ is realized by $q\le |S|$ pair-swaps $(x_t,y_t)$.
For swap $t$, let $\mathsf{path}_t(x)$ be the internal nodes on the path from the leaf storing $x$ to the root immediately before the swap.
Then the number of hashes recomputed by the incremental strategy is:
\[
O\left(\sum_{t=1}^{q}
\left|\mathsf{path}_t(x_t)\cup \mathsf{path}_t(y_t)\right|\right)
\le O(|S|h).
\]
Since any binary HuffMHT has $h\le n-1$, the worst-case update cost is $O(|S|n)$ hash computations.
\end{proposition}

\begin{proof}
Consider one pair-swap $(x_t,y_t)$.
All internal nodes outside $\mathsf{path}_t(x_t)\cup \mathsf{path}_t(y_t)$ have the same ordered child hashes before and after the swap, so their stored hashes remain valid.
Every internal node on these two root paths may receive a different child hash, and recomputing each such node once in bottom-up order restores all affected hashes.
The two leaf hashes add only a constant number of hash computations per swap and are absorbed in the asymptotic bound.
Thus swap $t$ costs
$O\left(\left|\mathsf{path}_t(x_t)\cup \mathsf{path}_t(y_t)\right|\right)$
hash computations.
Each root path contains at most $h$ internal nodes, so the union size is at most $2h$.
Summing over the $q\le |S|$ swaps gives $O(|S|h)$.
Finally, a binary tree with $n$ leaves and no unary internal nodes has height at most $n-1$, giving the stated worst-case bound.
\end{proof}

\begin{figure*}
    \centering
    \begin{center}
    \begin{tabular}{ccc}
    \scalebox{0.75}{
    \begin{tikzpicture}[
        x=1cm,y=1cm,
        >=latex,
        font=\small,
        hnode/.style={align=center, inner sep=1pt},
        leaf/.style={align=center, inner sep=1pt},
        swapleaf/.style={draw=red!70!black, very thick, rounded corners=2pt, align=center, inner sep=2pt, minimum width=9mm},
        dirty/.style={draw=orange!80!black, very thick, ->},
        clean/.style={draw=black, ->}
    ]
        \node[hnode] (R) at (0,3) {$H(H_1||H_{2,3,4})$};
        \node[hnode] (H1) at (-1.2,2) {$H_1 := H(m_1)$};
        \node[leaf] (m1) at (-1.2,1) {$m_1$};
        \node[hnode] (H234) at (2.2,2) {$H_{2,3,4} := H(H_2||H_{3,4})$};
        \node[hnode] (H2) at (0.7,1) {$H_2 := H(m_2)$};
        \node[swapleaf] (m2) at (0.7,0) {$m_2$};
        \node[hnode] (H34) at (4,1) {$H_{3,4} := H(H_3||H_4)$};
        \node[hnode] (H3) at (2.7,0) {$H_3 := H(m_3)$};
        \node[leaf] (m3) at  (2.7,-1) {$m_3$};
        \node[hnode] (H4) at    (5.3,0) {$H_4 := H(m_4)$};
        \node[swapleaf] (m4) at (5.3,-1) {$m_4$};

        \draw[clean] (H1) -- (R);
        \draw[clean] (H234) -- (R);
        \draw[clean] (m1) -- (H1);
        \draw[clean] (H2) -- (H234);
        \draw[clean] (m2) -- (H2);
        \draw[clean] (H34) -- (H234);
        \draw[clean] (H3) -- (H34);
        \draw[clean] (H4) -- (H34);
        \draw[clean] (m3) -- (H3);
        \draw[clean] (m4) -- (H4);

        \node[font=\footnotesize] at (1.5,-1.5) {Before: $m_2$ has the shorter path};
    \end{tikzpicture}
    }
    &
    $\Longrightarrow$
    &
    \scalebox{0.75}{
    \begin{tikzpicture}[
        x=1cm,y=1cm,
        >=latex,
        font=\footnotesize,
        hnode/.style={align=center, inner sep=1pt},
        leaf/.style={align=center, inner sep=1pt},
        swapleaf/.style={draw=red!70!black, very thick, rounded corners=2pt, align=center, inner sep=2pt, minimum width=9mm},
        dirty/.style={draw=orange!80!black, very thick, ->},
        clean/.style={draw=black, ->}
    ]
        \node[hnode] (R) at (0,3) {$H(H_1||H_{4,3,2})$};
        \node[hnode] (H1) at (-1.2,2) {$H_1 := H(m_1)$};
        \node[leaf] (m1) at (-1.2,1) {$m_1$};
        \node[hnode] (H432) at (2.2,2) {$H_{4,3,2} := H(H_4||H_{3,2})$};
        \node[hnode] (H4) at (0.7,1) {$H_4 := H(m_4)$};
        \node[swapleaf] (m4) at (0.7,0) {$m_4$};
        \node[hnode] (H32) at (4,1) {$H_{3,2} := H(H_3||H_2)$};
        \node[hnode] (H3) at (2.7,0) {$H_3 := H(m_3)$};
        \node[leaf] (m3) at  (2.7,-1) {$m_3$};
        \node[hnode] (H2) at    (5.3,0) {$H_2 := H(m_2)$};
        \node[swapleaf] (m2) at (5.3,-1) {$m_2$};

        \draw[clean] (H1) -- (R);
        \draw[dirty] (H432) -- (R);
        \draw[clean] (m1) -- (H1);
        \draw[dirty] (H4) -- (H432);
        \draw[dirty] (m4) -- (H4);
        \draw[dirty] (H32) -- (H432);
        \draw[clean] (H3) -- (H32);
        \draw[dirty] (H2) -- (H32);
        \draw[clean] (m3) -- (H3);
        \draw[dirty] (m2) -- (H2);

        \node[font=\footnotesize] at (1.5,-1.5) {After: $m_4$ moves to the shorter path};
    \end{tikzpicture}
    }
    \\
    (a) && (b)
    \end{tabular}
    \end{center}
    \captionDesc{A local pair-swap in HuffMHT. If $m_4$ is accessed more frequently than $m_2$, the update swaps their positions so that the path to $m_4$ is shorter. Only hashes on the highlighted ancestor paths are recomputed.}
    \label{fig:huffmht_incremental_swap}
\end{figure*}

\paragraph{Traditional Adaptive Huffman Coding}
A natural candidate for this setting is adaptive Huffman coding, which maintains a Huffman tree over a stream whose element set is not known in advance~\cite{knuth1985dynamic,vitter1987design}.
It keeps a not-yet-transmitted (NYT) leaf for unseen elements.
When a new element first appears, the NYT leaf is split into a leaf for the new element and a new NYT leaf.
When an existing element is accessed, the weight of its leaf is incremented.
After either case, the tree reshapes by moving nodes along the path from the updated node to the root, using swaps or slide-and-increment steps to restore the sibling property. 

This per-operation tree reshaping is the main obstacle to using adaptive Huffman directly as an ADS.
Each swap or slide changes ordered parent-child relations in the tree, so the hashes on the affected ancestor paths must be recomputed before the next authenticated root is committed.

\begin{restatable}[Adaptive Huffman Maintenance Cost]{proposition}{resAdaptiveHuffmanCost}
\label{res:adaptive-huffman-cost}
Let $\tau$ be an adaptive Huffman tree with $n$ leaves and height $h$ immediately before one access or insertion.
Let $A$ be the set of nodes whose authenticated leaf payload or ordered child relation changes during the update.
Let $D=\bigcup_{v\in A}\mathsf{path}(v)$, where $\mathsf{path}(v)$ contains $v$ and its ancestors up to the root.
Maintaining the authenticated root requires recomputing at most $O(|D|)$ hashes.
If the adaptive Huffman update performs $s$ local structural moves, where each move changes a constant number of ordered child relations, then $|D|\le O(\min\{(s+1)h,n\})$.
Thus the worst-case authenticated maintenance cost of one adaptive Huffman update is $O(n)$.
This bound is tight up to constants, since updating or inserting at a leaf of depth $h$ dirties the entire leaf-to-root path, and an adaptive Huffman tree can have height $h=n-1$.
\end{restatable}

For adaptive Huffman, authenticated maintenance cannot treat a swap as only a pointer update.
Once a leaf payload or an ordered parent-child relation changes, every ancestor hash whose input may contain that changed node must be refreshed before the next root commitment.
An update with $s$ structural moves therefore dirties the union of $O(s+1)$ root paths, each of length at most $h$, but never more than the $O(n)$ nodes in the tree.
The worst case is linear because adaptive Huffman trees can be unbalanced, so a single leaf update can dirty a path of length $n-1$.
The full proof is included in the appendix; see \Cref{sec:proofs}.

\paragraph{Continuous Insertions and Periodic Rebuild}
Our implementation keeps a base HuffMHT for elements that have already been absorbed into the optimized layout, and stores newly inserted or not-yet-absorbed elements in an overflow MT.
We refer to this as the \emph{periodic-rebuild} approach.
The base and overflow components are kept key-disjoint.
Elements stored in the overflow tree remain there until the next rebuild incorporates them into the base HuffMHT.
When an existing element is updated, the periodic-rebuild HuffMHT updates the leaf value and marks the ancestor path dirty, so the next root computation reflects the new value.
The layout changes are delayed until the rebuild boundaries.
When the overflow tree is non-empty, the periodic-rebuild HuffMHT root is computed as $H(\textsf{base-overflow}\Vert r_{\mathrm{base}}\Vert r_{\mathrm{overflow}})$, where $r_{\mathrm{base}}$ is the base HuffMHT root and $r_{\mathrm{overflow}}$ is the overflow Merkle root.
When the overflow MT is empty, the root is computed as $H(\textsf{base-only}\Vert r_{\mathrm{base}})$.
A proof against this root includes a component tag indicating whether the claimed key is authenticated by the base HuffMHT or by the overflow MT.
At a rebuild boundary, the periodic-rebuild HuffMHT merges the live base elements with the overflow elements, rebuilds the base HuffMHT from the latest frequency counts, and clears the overflow tree.
Thus, it handles dynamic insertions without reshaping the Huffman tree after every operation.
We evaluate this design choice in \cref{sec:dynamic-insert-microbenchmark}; \cref{sec:adaptive-huffman-replay-stress} reports a real-trace comparison with adaptive Huffman.

\section{Tiered Tree Architecture}
\label{sec:hmt-arch}

The intuition behind \dsname is that access patterns may be highly skewed, where a small and shifting subset of items are more frequently accessed than others over a given period of time.
To optimize update overhead and membership proof length for different access patterns, we store items in different tiers according to their popularity.
To account for changes in access frequencies, tier migration is applied deterministically at period boundaries before the next global root is committed.

In our evaluation, the system operates with two tiers:
\begin{itemize}
  \item \textbf{Cold Tier:} A balanced binary MT stores infrequent and newly inserted elements, providing logarithmic update paths in the cold tier size.
  \item \textbf{Hot Tier:} Select high-frequency elements are stored in a HuffMHT, reducing access-weighted proof size and path-rehashing work for those elements.
\end{itemize}

For such mutable workloads, a single frequency-optimized HuffMHT shortens paths for hot elements by placing colder elements deeper in the tree.
As a result, updates to cold elements can become expensive.
This motivates the two-tier instantiation above, where cold and newly inserted elements remain in the MT, while elements with sufficient observed access frequency migrate into the periodic-rebuild HuffMHT hot tier defined in \cref{sec:update-strategies}.

\paragraph{Why the Cold Tier Uses a Merkle Tree}
Cold and newly inserted elements have not yet accumulated enough accesses to benefit from a frequency-optimized layout.
Keeping them in a balanced binary MT avoids assigning them to the long tail of a Huffman-shaped tree.
A single HuffMHT gives short paths to hot elements by assigning longer paths to colder ones.
With $N$ leaves, the deepest HuffMHT leaf can have depth $N-1$, compared with $\lceil \log_2 N\rceil$ in a balanced binary MT.
Since inserts and value updates rehash the authenticated path to the root, this tail depth can make cold-element maintenance linear in the number of leaves.
\Cref{fig:hmt-mt-insert-ratio} plots the corresponding worst-case depth ratio using $\frac{N}{log_2 N}$.
HMT keeps cold and new elements in an MT to avoid this tail cost, and migrates elements into the HuffMHT only after their observed frequency warrants shorter hot tier paths.

\begin{proposition}[Worst-case Cold-element Depth Ratio]
\label{res:cold-item-depth-ratio}
For any binary HuffMHT with $N\ge2$ leaves, the maximum leaf depth is at most $N-1$.
Moreover, there exist access weights for which Huffman's algorithm produces a tree with a leaf at depth $N-1$.
In contrast, a balanced binary MT has maximum leaf depth $\lceil\log_2 N\rceil$.
Hence the worst-case authenticated-path depth ratio is
$\Theta\left(\frac{N}{\log N}\right).$
\end{proposition}

\begin{proof}
First consider any root-to-leaf path of depth $d$ in a binary tree with no unary internal nodes.
Each internal node on the path has a sibling subtree containing at least one leaf, and the path ends at one additional leaf.
Thus the tree has at least $d+1$ leaves, so $d\le N-1$.

Take Fibonacci-like access weights $1,1,2,3,5,\dots$, with tie-breaking chosen so that the previously merged subtree is merged with the next lightest leaf at each step.
Equivalently, an arbitrarily small perturbation can remove ties while preserving merge order.
With this order, each Huffman merge attaches one new leaf to the accumulated subtree.
One of the two initially lightest leaves is therefore placed below all $N-1$ internal nodes and has depth $N-1$.
Dividing this HuffMHT depth by the balanced MT depth $\lceil\log_2 N\rceil$ gives the stated $\Theta\left(\frac{N}{\log N}\right)$ ratio.
\end{proof}

\paragraph{Why Two Tiers}
Although \dsname can be generalized to more than two tiers, we use a two-tier instantiation because there is a tradeoff between tier-local path length and global membership-proof overhead.
Adding more tiers can place the hottest elements into a smaller ADS, which may shorten their local depth and result in a smaller proof.
However, the membership proof must include not only the local proof of the tier, but also the root of each tier to reconstruct the global root.
With more tiers, these additional roots increase the overall proof size and can outweigh the marginal reduction in tier-local path length for hot elements.
Since HuffMHT already places frequently accessed elements near the root within the hot tier, the marginal local-proof reduction from further splitting the hot set is limited.
Our tier-count sensitivity analysis in \cref{sec:tier-count-sensitivity} shows that, this additional multi-tier proof overhead compromises the potential hot-element savings; increasing the number of tiers does not provide further reduction in proof size.

Since the benefit of the hot tier requires tracking the most frequently accessed items, \dsname periodically reevaluates access frequencies and migrates items between tiers accordingly.
This adaptive migration is based on the following promotion and demotion rules, whose candidate selection relies on Promotion Cache metadata that tracks tier membership and frequency order for hot tier elements and cold tier promotion candidates, as described in \cref{sec:promotion-cache}.

\paragraph{Promotion}
\dsname monitors access frequencies.
If the highest-scoring lower-tier candidate exceeds a threshold, it becomes eligible for promotion to the next hotter tier.

\paragraph{Demotion}
\dsname supports movement from a hotter tier to a colder one.
The active migration policy determines when an element should leave its current hot tier.

\subsection{Membership-proof Soundness}
Tiering is internal to \dsname, while the interface described in \cref{ADS_Intro} is exposed externally.
We provide a high-level description of the interface in \cref{alg:hmt-ops}.
To allow one root to represent the entire state and maintain cryptographic integrity, the global root is calculated by concatenating the hashes of all individual root tiers and hashing the result.
Given a hash-function $H$ and $T$ root tiers $\left\{\mathrm{Root}_i\right\}_{i \in \{0,\dots,T-1\}}$, we define the global root $\mathrm{Root}_\mathrm{global}$:
\begin{center}
        $\mathrm{Root}_\mathrm{global} = H\left(\mathrm{Root}_0 \Vert \mathrm{Root}_1 \Vert \cdots \Vert \mathrm{Root}_{T-1}\right)$
\end{center}

We now show membership-proof soundness for \dsname.
Relative to an authenticated root $c$, no probabilistic polynomial-time adversary can output $(k,v,\Pi)$ such that $\mathsf{Verify}(c,k,v,\Pi)=1$ but $(k,v)$ is not represented by the committed tier roots, except with negligible probability.

We assume that $H$ is collision resistant and that all hashed inputs are prefix-free and domain separated, including leaf encodings, internal-node encodings, empty roots, tier-root tuples, and global-root tuples.
Throughout the algorithms, every displayed concatenation is shorthand for this canonical domain-separated encoding.
We also assume that the committed state is key-disjoint across tiers, so at any committed root, each live key appears in exactly one tier.
The adversary is given the public algorithms, tier policy, public metadata, and commitment $c$, and may choose any purported value and proof.
The committed dataset is the state represented by the honest tier roots at the time $c$ is produced.

\begin{proposition}[Multi-tier Membership Soundness]
\label{resultMultiTierMembership}
Let $D_0,\dots,D_{T-1}$ be key-disjoint tier datasets, and let each tier $i$ be represented by a membership-sound ADS $\mathcal{A}_i$ with root $r_i=\mathcal{A}_i.\mathsf{Commit}(D_i)$.
Let the global commitment be the hash of the ordered tuple $(r_0,\dots,r_{T-1})$; in \cref{alg:hmt-ops}, the displayed concatenation denotes this fixed tuple encoding.
So the resulting ADS is membership sound.
\end{proposition}

\begin{proof}
Suppose an adversary produces a false accepting proof for $(k,v)\notin \bigcup_i D_i$, denoted by $(\hat r_0,\dots,\hat r_{T-1})$, and let the honest proof by $(r_0^\star,\dots,r_{T-1}^\star)$.
If they differ, acceptance gives two distinct inputs with the same hash, contradicting collision resistance.
Otherwise, the local proof verifies against the honest root $r_i^\star$, although $(k,v)\notin D_i$, contradicting membership tier $i$'s soundness.
\end{proof}

For the sake of completeness, we also provide \cref{resultHuffMembership}.
\begin{proposition}[HuffMHT Membership Soundness]
\label{resultHuffMembership}
A HuffMHT with key-value leaves and ordered-child internal hashes is membership sound when $H$ is collision resistant.
\end{proposition}

\begin{proof}
The proof is similar to the canonical MT argument.
A false accepting proof reconstructs the honest root from a claimed leaf and an ordered sequence of siblings.
Comparing the reconstructed path with the honest tree, the first point at which different encoded inputs yield the same digest gives a collision in $H$.
The argument is independent of balance because Huffman layout changes only path lengths, not the bottom-up authentication recurrence.
\end{proof}

\begin{proposition}[\dsname Membership Soundness]
The \dsname construction is membership sound when its constituent tier ADSs are membership sound and $H$ is collision resistant.
\end{proposition}
\begin{proof}
By \cref{resultHuffMembership}, HuffMHTs are membership sound; ordinary Merkle trees are membership sound by the same argument.
A periodic-rebuild HuffMHT is membership sound when its exported root is a domain-separated commitment to its mode and constituent roots.
If the overflow tree is empty, the exported root commits to the base HuffMHT root in the base-only domain.
If the overflow tree is non-empty, the exported root commits to the ordered pair consisting of the base HuffMHT root and the overflow Merkle root, and a proof identifies which component contains the claimed key.
In either case, a false accepting proof yields either a false proof for the base HuffMHT, a false proof for the overflow MT, or a collision in the exported-root hash.
Applying \cref{resultMultiTierMembership} to the resulting tier roots gives membership soundness of the global \dsname commitment.
The CMS, promotion cache, and migration policy only determine which tier stores each key before the next root is committed; they do not change verification once the tier roots are fixed.
\end{proof}

\subsection{Frequency Tracking via Count-Min Sketch (CMS)}

To dynamically optimize layout, \dsname maintains access-frequency estimates for when items should move between tiers.
Instead of exact per-item counters, we use Count-Min Sketch (CMS)~\cite{cormode2005countmin}, a $d \times w$ counter array.
With $w=\lceil \frac{e}{\epsilon}\rceil$ and $(d=\lceil \ln(\frac{1}{\delta})\rceil$, CMS uses $O(wd)$ counters.
For a fixed queried key $x$, its estimate $\hat f(x)$ satisfies $f(x)\le \hat f(x)\le f(x)+\epsilon T$ with probability at least $1-\delta$, where $f(x)$ is the true frequency of $x$ and $T$ is the total observed accesses.
This guarantee is per queried key, rather than over all observed keys at once.

E.g., with $\epsilon=\delta=10^{-6}$ and 32-bit (4-byte) counters, CMS uses $4\lceil \frac{e}{\epsilon}\rceil\lceil\ln(\frac{1}{\delta})\rceil\approx 152$ MB, and the per-query overestimation bound is at most $10^3$ accesses when $T=10^9$.
In contrast, exact counting for $10^7$ distinct 32-byte keys needs at least 360 MB for keys and counters, excluding hash-table overheads.
So, after $(\epsilon,\delta)$ are chosen, the CMS memory budget is fixed rather than proportional to the number of distinct elements.
\dsname derives CMS row indices from one BLAKE3 hash computation~\cite{oconnor2020blake3}, and uses conservative update to reduce overestimation in practice~\cite{goyal2012sketch,benmazziane2022cmscu}.
The standard CMS bound assumes suitable row hash functions; our deterministic hash instantiation is an engineering choice so replicas derive the same metadata from the same access stream.
A CMS update or query touches one counter in each of the $d$ rows, so CMS maintenance costs $O(d)$ time per operation.
\subsection{Promotion Cache} \label{sec:promotion-cache}

To coordinate element migration, \dsname uses an in-memory \emph{Promotion Cache} as policy metadata, not as authenticated state.
For each tracked key, the Promotion Cache records its current tier and maintains the frequency order used for promotion and demotion decisions.

For a hot tier with capacity $C_i$, the corresponding Promotion Cache also has capacity $C_i$; every key in the hot tier ADS has a cache entry.
Cold tier cache entries, by contrast, are only promotion candidates and may be evicted from metadata; hot tier entries record exact tier membership and must remain until the key is demoted or deleted.
The cached frequency order identifies candidates for promotion to a hotter tier and cold keys for demotion.
The cold tier may be much larger, so its Promotion Cache does not mirror the full cold tier ADS; it tracks only possible promotion candidates.
Thus, CMS estimates the frequency of a queried key, and the Promotion Cache determines which keys are currently eligible for migration.
CMS does not enumerate hot keys by itself; promotion and demotion candidates are selected only from the bounded set of keys tracked by the Promotion Cache.
This means cold keys absent from the cache may be missed until they are observed again, and CMS overestimation can affect layout quality but not proof soundness.
Promotion and demotion update both the affected ADS tiers and the corresponding Promotion Cache entries.

We implement each tier's cache by batching keys into frequency ranges of width $R$, with the replacement policy following the least-recently-used (LRU) method \cite{maffeis1993cache}.
This reduces the number of bucket relocations and avoids moving keys after every frequency change.

All Promotion Cache instances share one CMS for frequency estimation. 
\textsc{Locate} scans the $L$ caches and performs expected $O(1)$ key lookup in each, so it costs $O(L)$.
In our evaluated configurations, $L$ is fixed and small, so this lookup component is constant.

When a key is accessed, \dsname updates its CMS frequency estimate in $O(d)$. 
If the key remains in the same frequency bucket, the update is $O(1)$; otherwise, finding or creating the destination bucket costs $O(K_i)$ in the worst case, where $K_i$ is the number of active buckets in tier $i$. 
Eviction is $O(1)$ by removing the least-recent key from the lowest-frequency bucket, given maintained pointers to the minimum-frequency bucket.  
\Cref{alg:lfu} summarizes the Promotion Cache operations used to locate, update, and migrate elements.

The total space is $O(wd+\sum_i(N_i+K_i))$, where $O(wd)$ is the shared CMS space, $N_i\le C_i$ is the number of cached keys in tier $i$, and $K_i$ is the number of active buckets.
\section{Dynamic Customization}
\label{sec:dynamic}
HMT dynamically adjusts tiering periodically per batch of accesses, where batch (period) length is a tunable parameter.
At each period boundary, migration rules are applied given three policy choices (summarized in \cref{tab:migration_policy_choices}): how to score promotion candidates, how often to evaluate migration, and whether admission parameters remain fixed or feedback-controlled.
The resulting implemented variants are summarized separately in \cref{tab:migration_policy_variants}.
We number tiers from coldest to hottest as $0,\dots,L-1$. 
For adjacent tiers $i$ and $i+1$, promotion edge $i$ denotes the tier boundary that an element crosses when it moves from tier $i$ to the hotter tier $i+1$.
\begin{table}
\centering
\captionDesc{{Policy choices for tier migration.}}
\small
\begin{tabular}{m{5.2em}m{10em}m{10em}}
\toprule
{\textbf{Choice}}
&
{\textbf{What it controls}}
&
{\textbf{Options}}
\\ \midrule
{\textbf{Promotion score}}
&
{Which elements are considered hot.}
&
{Based on cumulative frequency, lifetime rate, or recent-window rate.}
\\ \midrule
{\textbf{Evaluation schedule}}
&
{When promotion and demotion are checked.}
&
{Every $I$ periods, for different $I$ values.}
\\ \midrule
{\textbf{Parameter control}}
&
{Whether hot tier threshold and capacity change over time.}
&
{Fixed or feedback-controlled.}
\\ \bottomrule
\end{tabular}
\label{tab:migration_policy_choices}
\end{table}

\begin{table}
\centering
\captionDesc{Implemented tier-migration policy variants.}
\begingroup
\footnotesize
\setlength{\tabcolsep}{2pt}
\renewcommand{\arraystretch}{1.14}
\begin{tabular}{@{}p{0.16\columnwidth}p{0.23\columnwidth}p{0.17\columnwidth}p{0.19\columnwidth}p{0.15\columnwidth}@{}}
\toprule
\textbf{Variant}
&
\textbf{Promotion score}
&
\textbf{Evaluation}
&
\textbf{Control}
&
\textbf{Main tradeoff}
\\ \midrule
\textbf{Absolute-Threshold}
&
Cumulative access count from the CMS.
&
Every period.
&
Fixed threshold and capacity.
&
Simple; history never decays.
\\ \midrule
\textbf{Ratio-Based}
&
Lifetime access rate.
&
Every period.
&
Fixed threshold and capacity.
&
Stable; new bursts may lag.
\\ \midrule
\textbf{Periodic-Evaluation}
&
Lifetime access rate.
&
Every configured number of periods.
&
Fixed threshold and capacity.
&
Fewer checks; slower response.
\\ \midrule
\textbf{Sliding-Window}
&
Recent-window access rate.
&
Every period.
&
Fixed threshold with delayed demotion.
&
Tracks recent hot-set shifts.
\\ \midrule
\textbf{Dynamic-Control}
&
CMS-based lifetime access rate.
&
Promotion every period; review feedback periodically.
&
Feedback-controlled threshold and capacity.
&
Adapts admission and tier size.
\\ \bottomrule
\end{tabular}
\endgroup
\label{tab:migration_policy_variants}
\end{table}

\subsection{Absolute-Threshold}
For promotion edge $i$, Absolute-Threshold scores each candidate $x$ tracked by the Promotion Cache for tier $i$ according to its cumulative CMS estimate $\hat f(x)$.
An item is eligible for promotion when $\hat f(x) > \theta_i$, where $\theta_i$ is the configured minimum cumulative frequency required for promotion from tier $i$ to $i+1$.
If the destination hot tier has room, $x$ is inserted there and removed from the colder tier.
If the hot tier is full, \dsname compares the candidate with the least-frequent element already admitted to that tier. 
The swap happens only when the candidate's cumulative estimate is strictly larger. 
A potential drawback with this design is that counts never decay. 
An element that was hot early in the run can keep a large score after its activity drops, while a newly hot element may need many periods to catch up, which can be undesirable when workloads change rapidly.
We provide the pseudocode in \cref{alg:absolute-threshold} with more details.

\subsection{Ratio-Based}

Ratio-Based uses a lifetime access-rate score. 
For promotion edge $i$ and number of elapsed periods $B$, it scores each Promotion Cache candidate $x$ by $s_B(x)=\frac{\hat f(x)}{B}$.
A candidate is eligible for promotion when $s_B(x)>\theta_i$, where $\theta_i$ is the configured minimum lifetime access rate for promotion across edge $i$.
Because candidates considered at a given moment share the same denominator $B$, Ratio-Based ranks candidates in the same order as Absolute-Threshold; equivalently, the eligibility test is $\hat f(x)>B\theta_i$.
This normalization makes the promotion threshold grow with run length, but the policy relies on lifetime counts, so stale hot keys may remain favored and newly hot keys may take time to become promotion candidates.

\subsection{Sliding-Window}

Sliding-Window \dsname targets workloads where hot-set changes are expected over short time scales.
Unlike the lifetime-frequency policies, it does not use cumulative CMS estimates for promotion, since CMS counters do not expire accesses that leave the window.
Instead, \dsname maintains a ring of recent period histograms and an aggregate window histogram $H_{\mathcal{W}}$; between periods, it removes the oldest period histogram and adds the newest one.

Let the number of periods included in the sliding window be $W$.
For element $x$, let $H_{\mathcal{W}}(x)$ be its frequency in the current window $\mathcal{W}$, and define:
$s_W(x)=\frac{H_{\mathcal{W}}(x)}{W}$.
For promotion edge $i$, a candidate $x$ is eligible for promotion when $s_W(x)\ge\theta_i$, where $\theta_i$ is the configured minimum window-local access rate for promotion across edge $i$.
Because accesses outside the window are removed from $H_{\mathcal{W}}$, stale hot items lose influence once they stop appearing in recent periods.
A hot-tier item is not demoted immediately when its window-local rate falls below the threshold.
The implementation places the stale hot element on a per-edge demotion wheel and rechecks it after a delay; demotion occurs only if the item still $s_W(x)<\theta_i$ when its wheel slot expires.
This lazy demotion policy reduces churn for items oscillating around the boundary while preserving responsiveness to bursts.
\Cref{alg:sliding-window} gives the pseudocode with more details on how it works.

\subsection{Dynamic-Control}

Dynamic-Control \dsname removes the need to choose one static promotion threshold and one static hot tier capacity for all workload phases. 
For promotion edge $i$, we maintain a threshold $\theta_i$, hot tier capacity $C_i$, and two feedback loops.

\paragraph{Threshold Control}
Every $M$ periods, the controller checks occupancy $\rho_i=\frac{n_i}{C_i}$ (where $n_i$ is the current hot tier size) and rejection ratio $q_i=\frac{R_i}{A_i}$ (where $A_i$ and $R_i$ are the promotion attempts and rejected attempts in the window). 
The terms ``high'' and ``low'' are policy choices: one can configure an occupancy band $[\rho_{\mathrm{low}},\rho_{\mathrm{high}}]$ and a rejection threshold $q_{\mathrm{high}}$. 
If $\rho_i>\rho_{\mathrm{high}}$ and $q_i>q_{\mathrm{high}}$, the controller raises $\theta_i$ to admit fewer candidates and reduce churn. 
If $\rho_i<\rho_{\mathrm{low}}$ for $U$ consecutive control windows, it lowers $\theta_i$.
The threshold is bounded by $[\theta_i^{\min},\theta_i^{\max}]$, where $\theta_i^{\min}$ is configured and $\theta_i^{\max}$ is the edge-$i$ upper bound derived from the least-frequent element currently admitted to the destination tier $i+1$.
This upper bound prevents the controller from setting the promotion threshold above the access rate at the tail of tier $i+1$. 
Thus, a candidate in tier $i$ that is hotter than the current tail of tier $i+1$ can still be considered for promotion, allowing stale elements in tier $i+1$ to be replaced by more frequently accessed candidates.

\paragraph{Capacity Control}
A loop evaluates the hot tier capacity in every $K$ control window, i.e., every $KM$ periods. 
The controller increases $C_i$ if the hot tier remains congested while $\theta_i$ is near $\theta_i^{\max}$, and shrinks $C_i$ if the tier is under-utilized while $\theta_i$ is near $\theta_i^{\min}$, so capacity changes only after threshold adjustment is exhausted for the current hot-tier configuration.

The two loops run at different cadences. 
Threshold control reacts to short-term admission pressure, whereas capacity control adjusts the hot tier capacity only after that pressure persists across multiple threshold-control windows.

We report the parameter values used in our experiments in \cref{sec:evaluation}, but the design is not tied to those numeric choices.
For replicated deployments, the parameters are fixed for a run, and controller updates occur only at period boundaries.
Given the same access stream and deterministic tie-breaking, replicas therefore make the same promotion and demotion decisions.
\Cref{alg:dynamic-control} gives the pseudocode with more details.

\section{Evaluation}
\label{sec:evaluation}
We evaluate \dsname using two types of benchmarks.
First, controlled microbenchmarks evaluate the layout-update strategies introduced in \cref{sec:incremental-update}.
They measure the maintenance cost of rebuilding and updating Huffman-authenticated layouts under skewed workloads.
Second, we employ a deterministic replay of real Ethereum state accesses at account granularity, where each Ethereum account is treated as one authenticated element, and every compared ADS processes the same operation trace in the same order.

\subsection{Layout-Update Microbenchmarks}
\label{sec:layout-update-microbenchmarks}

\subsubsection{Continuous Updates}
We evaluate the incremental update strategy of \cref{sec:incremental-update} by comparing two HuffMHT instances over the same key set: one performs a full rebuild at each batch boundary, while the other batches accesses between boundaries and applies frequency-order pair swaps.

We first insert $N=100{,}000$ keys into both instances, and then apply $5{,}000{,}000$ get/update operations drawn from a Zipf distribution with $a=4$, $b=1$.
This removes insertion cost from the comparison, so the measured cost at each batch boundary comes only from adapting the tree shape to changed access frequencies.
At every $500{,}000$-operation batch boundary, we compare rebuilding the Huffman tree from current frequencies against applying frequency-order pair swaps.
\Cref{fig:zipf_milestone_updates} reports the per-batch restructuring time.
Batch 1 is used only to rebuild both instances from the frequencies observed during the first $500{,}000$ operations.
The full-rebuild and batched pair-swap strategies are compared from batch 2 onward.
Over batches 2--10, swaps average $47.1$ ms, while full rebuilds average $100.3$ ms, a $2.13\times$ reduction.
This shows that incremental restructuring can reduce rebuild overhead for stable key sets.

\begin{figure}
    \centering
    \begin{minipage}{0.48\columnwidth}
        \centering
        \includegraphics[width=\linewidth]{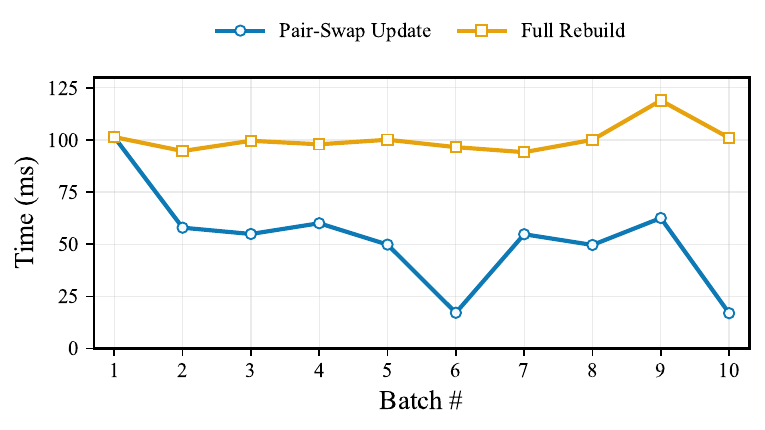}
        \captionDesc{Restructuring times for a stream of batched-accesses which induce frequency changes over a fix set of items. 
        (y-axis Batch\# instead of Milestone (also in the main text) (Done)}
        \label{fig:zipf_milestone_updates}
    \end{minipage}
    \hfill
    \begin{minipage}{0.48\columnwidth}
        \centering
        \includegraphics[width=\linewidth]{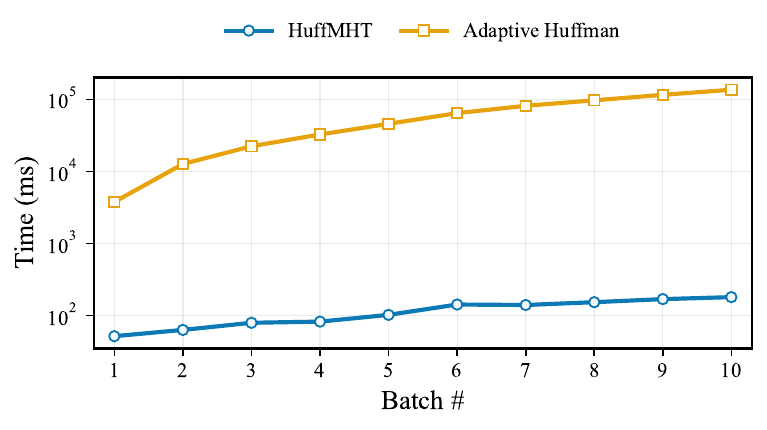}
        \captionDesc{Update times for a stream of batched-accesses which induce both frequency changes and new insertions.}
        \label{fig:dynamic_insert_maintenance}
    \end{minipage}
\end{figure}

\subsubsection{Continuous Inserts}
\label{sec:dynamic-insert-microbenchmark}
We next evaluate update time under continuous insertions.
The workload starts from an empty tree and inserts exactly $N=100{,}000$ keys over $5{,}000{,}000$ operations.
The first operation is forced to be an insert, and the remaining insert positions are randomly distributed across the run.
Non-insert operations (get/update) access already-inserted keys under a Zipf distribution with $a=4$, $b=1$.
We compare periodic-rebuild HuffMHT, which inserts new keys into an overflow MT and rebuilds at each $500{,}000$-operation batch boundary, against an Adaptive Huffman tree, which updates its tree shape after every insert, get, or update.
\Cref{fig:dynamic_insert_maintenance} reports the update time measured for every batch of 500,000-operations.
For periodic-rebuild HuffMHT, this time includes the operations performed during the batch, plus the rebuild time at the batch boundary.
For Adaptive Huffman, this time includes the adaptive tree update performed after each operation during the same batch.
At the final batch, periodic-rebuild HuffMHT spends about $180$ ms, while Adaptive Huffman spends about $136$ seconds.
This gap comes from the cost of maintaining an authenticated Huffman layout after every operation.
Periodic rebuilds avoid this per-operation restructuring by putting new elements in the overflow MT and rebuilding the HuffMHT only at batch boundaries.

\subsection{Ethereum Replay Evaluation}
\label{sec:ethereum-replay-evaluation}
\subsubsection{Experimental Setup}

\paragraph{Data}
The input is a trace of go-ethereum \texttt{StateDB} operations recorded during block execution.
All compared ADSs process this trace in the same order.
The replayer follows \texttt{snapshot}/\texttt{revertsnapshot} markers so reverted operations do not affect ADS state.

Every Ethereum account is mapped to one authenticated element.
The trace records 20-byte Ethereum account addresses; before lookup or update, the replayer maps each address to the same 32-byte account key used by Ethereum's account trie.
Storage operations are attributed to the owning contract account.
A \texttt{getstate} contributes one read access to that account, and a \texttt{setstate} contributes one update access with the recorded storage value used as deterministic payload.
At each \texttt{commit} marker, the replayer recomputes each ADS root and records the per-block total hash input and weighted average proof size for the block.

The replay starts at block 17{,}000{,}000 without reconstructing the full Ethereum pre-state.
On first access to an account, every ADS materializes it with the same deterministic placeholder before applying the recorded operation. 
Thus, all ADS variants process the same account-level operations, in the same order, with deterministic payloads.

\paragraph{Compared Data Structures}
We compare the following:
\begin{itemize}
    \item \textbf{MPT}: an Ethereum-style Merkle Patricia Trie baseline.
    \item \textbf{UBT}: a unified binary authenticated tree baseline, representing the binary-tree direction considered for Ethereum state~\cite{buterin2025eip7864}.
    \item \textbf{Ratio-Based \dsname}: a two-tier configuration, where the hot and cold tiers are respectively instantiated using a HuffMHT and a MT.
    Promotion uses lifetime access rate, with threshold $0.05$ and cold/hot promotion cache capacities of $8{,}000/16{,}000$ entries.
    \item \textbf{Sliding-Window \dsname}: the same two-tier configuration, but promotion uses only recent accesses in a $1{,}000$-block rolling window. 
    It uses a promotion-rate threshold of 0.05 and a 100-block lazy-demotion delay.
    \item \textbf{Dynamic-Control \dsname}: starts from Ratio-Based and uses feedback to adjust hot-tier capacity and promotion threshold. 
    The controller uses $M=2{,}500$-block threshold-control windows, evaluates capacity every $K=3$ windows, targets occupancy band $[0.75,0.95]$, uses rejection threshold $q_{\mathrm{high}}=0.3$, lowers thresholds after $U=3$ low-occupancy windows, and clamps thresholds with $\theta_i^{\min}=0.05$ and tail-derived $\theta_i^{\max}$. 
\end{itemize}
Each \dsname configuration rebuilds its HuffMHT every 500 blocks and uses LFU Promotion Cache buckets of range 10.
\Cref{sec:additional-policy-replay} reports two additional HMT variants.
\Cref{sec:underlying-ads-choice} evaluates alternative underlying ADS choices under the same replay setup.

\paragraph{Metrics}

We report two per-block metrics. 
Different ADSs in our evaluation store different amounts of information per node. 
When they traverse their respective data structures to generate proofs, verify proof paths, or update the data structure, the computation cost is determined by the amount of data that is input to hashing.
First, hashed input bytes measures the total number of bytes hashed during a block commit. 
We compute this quantity per authenticated node, using the byte string included in that node's hash input: MPT contributes trie-node bytes; Merkle and HMT leaves contribute a domain separator, key, and value payload, while internal nodes contribute a domain separator and child hashes; UBT contributes the inputs to its stem and internal-node hashes. 
For \dsname, we sum this quantity across tiers.  
Weighted average proof size measures the inclusion-proof size for keys accessed in a block, weighted by each key's access count in that block.
\Cref{sec:read-write-proof-size} records this average per-block separately for read and write accesses.

\begin{figure*}
    \centering
    \includegraphics[width=0.48\columnwidth]{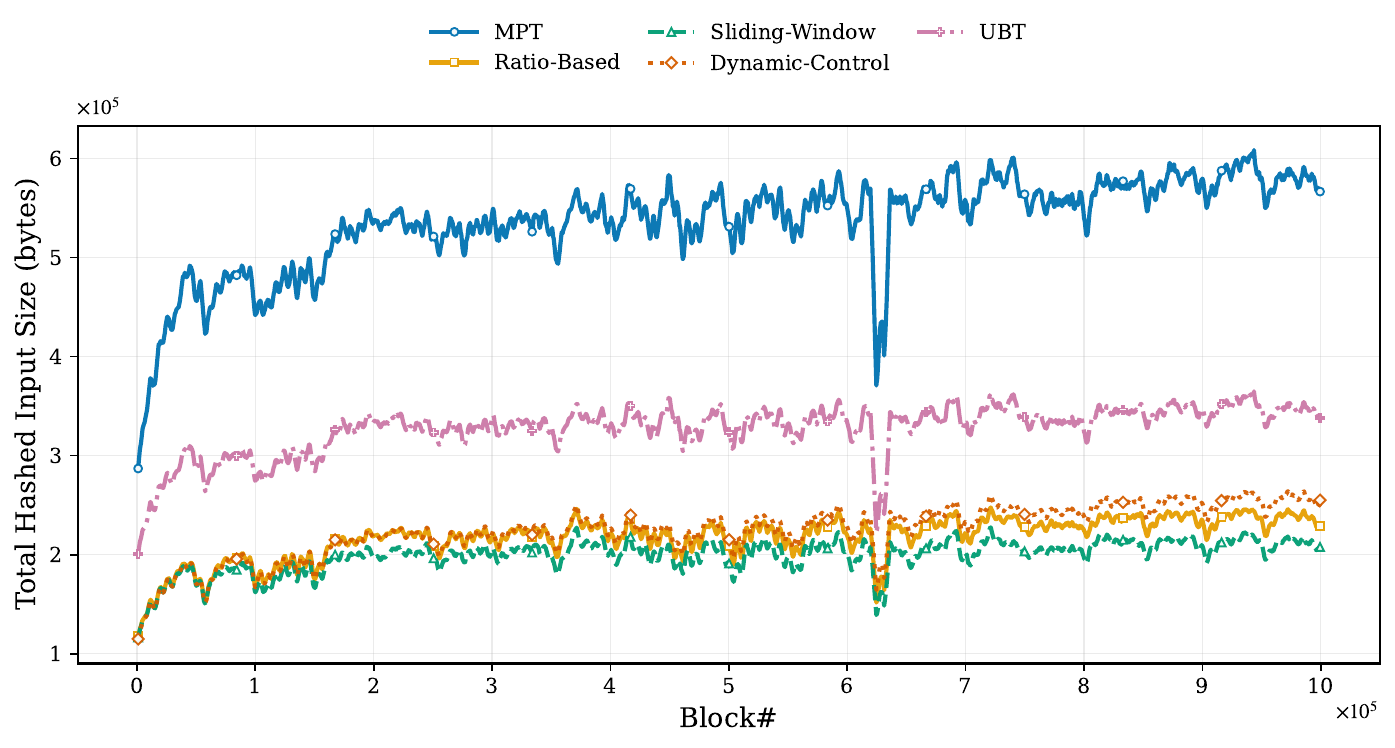}
    \hfill
    \includegraphics[width=0.48\columnwidth]{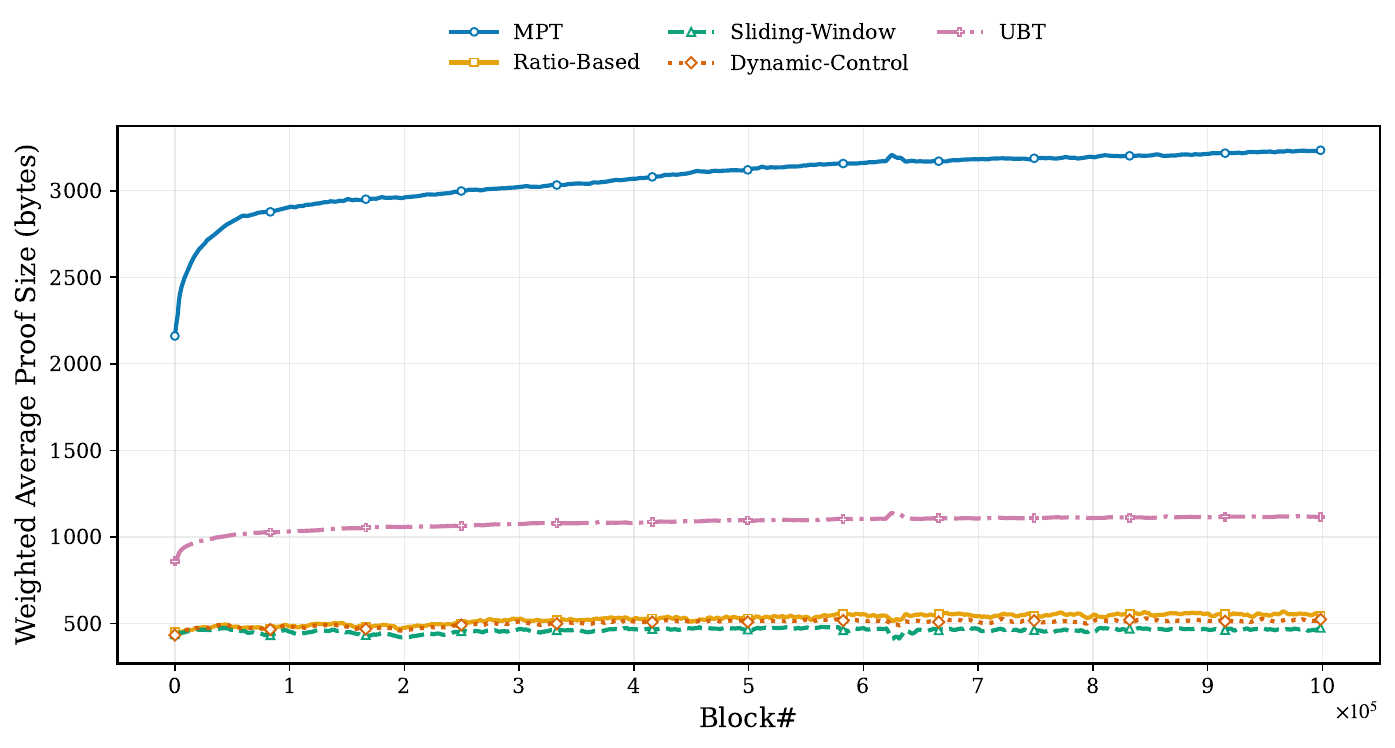}
    \vspace{-0.4em}
    \par
    \makebox[0.48\columnwidth][c]{\small\phantomsection\label{fig:eth_replay_hash_input_bytes:a}(a)}
    \hfill
    \makebox[0.48\columnwidth][c]{\small\phantomsection\label{fig:eth_replay_hash_input_bytes:b}(b)}
    \captionDesc{Performance comparisons of MPT, UBT, and HMT instantiations with various policies. The evaluation (fully described in \cref{sec:ethereum-replay-evaluation}) uses one million Ethereum blocks. Fig.~(a) shows the total bytes hashed during the processing of each block, and Fig.~(b) shows the weighted average inclusion proof size.}
    \label{fig:eth_replay_hash_input_bytes}
\end{figure*}

\paragraph{Testbed}
We run all experiments on a Dell PowerEdge R7715 server with a 16-core AMD EPYC 9135 CPU and 32~GB of DRAM.
Our Go implementation is compiled with Go 1.25.4.

\subsubsection{Evaluation Results}
\paragraph{Hash Input Bytes}
\Cref{fig:eth_replay_hash_input_bytes}\hyperref[fig:eth_replay_hash_input_bytes:a]{(a)} shows the per-block hashed input bytes across the one-million-block replay.  
The x-axis reports the replay offset from block 17{,}000{,}000, and the y-axis reports the total number of bytes input to hash functions during the commit for that block.
The visible dip around replay offset $6.2\times 10^5$ is workload-driven rather than specific to one ADS. 
A 10{,}000-block summary of the trace shows that the average number of distinct updated account keys drops from about $68$ per block to about $48$ per block in this interval, even though total operation volume does not drop proportionally.
Since repeated writes to the same account coalesce into the same dirty path at commit time, all ADSs perform less rehashing in this region.
MPT has the largest byte-level hash workload throughout, while UBT remains the second most expensive method.
The \dsname variants consistently outperform both: Ratio-Based and Dynamic-Control stay close to one another, and Sliding-Window is the lowest-cost method for most blocks.

\paragraph{Proof Size}
\Cref{fig:eth_replay_hash_input_bytes}\hyperref[fig:eth_replay_hash_input_bytes:b]{(b)} shows the weighted average proof size for the accounts accessed in each block. 
The x-axis reports the replay offset from block 17{,}000{,}000, and the y-axis reports the inclusion-proof size averaged over the keys accessed in that block, weighted by the number of times each key appears.
MPT has the largest average proof size, around several kilobytes per accessed account. 
UBT reduces this cost by using a binary authenticated layout, but its proofs remain roughly twice as large as the \dsname variants. 
The \dsname variants keep frequently accessed accounts in the hot HMT tier, so the access-weighted average is dominated by short hot tier paths rather than by cold tier paths. 
Sliding-Window is lowest or tied for lowest over most of the replay.

\paragraph{Policy Tradeoffs}
Across the hashed-input and proof-size figures, the three \dsname policies expose different tradeoffs.  
Ratio-Based is simple and stable, but lifetime frequency can keep older hot accounts influential after the workload shifts. 
Sliding-Window reacts to recent locality and is the best policy in these figures for both hashed input bytes and weighted proof size.  
Dynamic-Control remains competitive in proof size but consumes more hashed input bytes than the static \dsname policies in this run. 
We therefore treat Dynamic-Control as an exploratory adaptive policy rather than the primary winning configuration.

\section{Additional Related Literature}
\label{sec:related}
We now go over additional relevant work.

\paragraph{Blockchain State Commitments}
Real-world performance considerations impose constraints on ADS-based solutions for blockchain state commitments: proof-sizes should be compact, and reading and updating operations should be performant.
Ethereum's MPT has large proofs, and its performance is noted as hampering scaling~\cite{wood2025ethereum}. 
Proposals to use Verkle Trees (VTs) for blockchain state were motivated by stateless validation, where vector-commitment-based witnesses can be much smaller than MPT branches~\cite{ballet2024eip7748,buterin2024verge}.
More recent Ethereum state-tree discussions keep the stateless-verification goal but make the choice of commitment structure more open-ended, emphasizing post-quantum security and compatibility with SNARK/STARK-based validity proofs; these considerations motivate a unified binary tree rather than VTs, motivating our focus on tree-based structures~\cite{buterin2024verge,buterin2025eip7864}.

Other commitment schemes target stateless validation and efficient proof updates using accumulators, key-value commitments, and vector commitments~\cite{boneh2019batching,agrawal2020kvac,tas2023vector}. 
Blockchain-specific ADSs have also optimized for query expressiveness or tree height. 
GEM$^2$-tree supports authenticated range queries in hybrid-storage blockchains while reducing smart-contract maintenance gas through merge-oriented authenticated structures~\cite{zhang2019gem}. 
MHOT replaces fixed-prefix trie traversal with discriminative-bit indexing and hierarchical Merkle proofs to reduce height, proof size, and vulnerability to Nurgle-style prefix attacks while remaining hash-based~\cite{xie2026mhot}. 
Storage-oriented authenticated-state systems such as mLSM, RainBlock/DSM-Tree, Jellyfish Merkle Tree, LVMT, QMDB, and AlDBaran instead focus on improving storage locality, reducing I/O, lowering update cost, or simplifying state management~\cite{raju2018mlsm,ponnapalli2021rainblock,GaoJellyfishMT,li2024lvmt,zhang2025qmdb,kauer2025aldbaran}. 
One delayed-write state-tree design buffers recently updated state in a sliding-window authenticated cache before committing those updates to the global state tree, reducing state-root latency while remaining compatible with MPT and VT~\cite{sheng2026swmt}.

In contrast to previous work, HMT does not change the commitment primitive itself (unlike, e.g., Verkle and vector-commitment proposals).
Instead, it targets the data-structure's layout rather than primarily the backing store or I/O path (unlike storage-oriented designs, for example), and does so by adapting short authenticated paths to keys that are hot under the observed workload while preserving a hash-based state commitment (unlike height-optimized or gas-aware state trees).

\paragraph{Frequency-aware Trees}
Some tree-based ADSs opt for skewed constructions~\cite{karpinski2004note} suitable to various settings, e.g., certificate revocation~\cite{elwailly2004quasimodo} and memory integrity checking~\cite{szefer2014towards}.
As we show, of particular relevance to blockchains is a specific design where frequently accessed elements are positioned closer to the root, because this can minimize the overhead of updating hot keys and the overall expected proof length.
A natural approach for constructing such trees is to use Huffman coding, which is optimal in some cases~\cite{moffat2019huffman} and can obtain a substantial improvement over vanilla MTs for certain access distributions.
However, Huffman optimality is achieved only when element accesses follow a static distribution and are independent and identically distributed, while blockchains are dynamic and may also feature intricate access patterns (e.g., it is reasonable to expect that some decentralized exchanges are more likely to be accessed in tandem than others).
Indeed, devising performant ADSs for the case where access distributions are dynamic is highlighted as a promising future direction by previous work~\cite{szefer2014towards}, thus motivating our study.

\paragraph{Other Authenticated Data Structures}
Authenticated dictionaries and hash tables provide the basic abstraction of committing to a mutable map and proving query answers~\cite{crosby2011authenticated,papamanthou2008authenticated,miller2014authenticated}. 
A recent systematization of authenticated dictionaries unifies their deployment models, security definitions, and construction families, and highlights a persistent lookup/update tradeoff in known schemes~\cite{malvai2026sok}.
Complementary theoretical work formalizes black-box relationships between authenticated dictionaries, authenticated sets, and memory checkers, showing that accumulator-based black-box constructions incur inherent lookup and update overheads~\cite{falzon2025blackbox}. 
Recent memory-checking work studies locality-optimized layouts for authenticated storage~\cite{wang2023locality}.

\paragraph{Other Frequency-aware ADSs}
Frequency-aware authenticated dictionaries have used non-uniform query models and skewed layouts to reduce average proof size under Zipf-like workloads~\cite{atighehchi2015new}. 
A related line of work in oblivious maps separates hot and cold records to improve hot-query response time, either while preserving full obliviousness or by exposing limited tier-membership leakage~\cite{li2026tieredomap}.
Prior frequency-aware authenticated dictionaries and skewness-aware OMAPs therefore support the broader premise that access skew is useful, but they do not provide the blockchain-oriented dynamic authenticated-map design we need: frequent root commits, insertions and deletions, and deterministic layout.

\paragraph{Evaluating Dynamic Data Structures}
A common approach to evaluating dynamic data structure algorithm performance (e.g., insertion of a new key, or modifying an existing key's value) considers an online setting \cite{sleator1985amortized, awerbuch1992competitive, aspnes1998competitive}.
Requests arrive over time, and the examined algorithm must complete them without having any reliable forecast of future workload \cite{hajek2001competitiveness, gafni2022greedy, gafni2024scheduling}.
One benchmark used by the theoretic literature is the \emph{competitive ratio}, i.e., the performance of a given online algorithm facing adversarially-generated requests is compared to an optimal offline one which has full knowledge of the future sequence of requests \cite{borodin2005online}.
While this type of analysis adopts a worst-case perspective, our focus is on optimizing real-world average-case performance.

\section{Conclusion}
\label{sec:conclusion}
This paper addresses the cost of authenticated state access in dynamic, skewed workloads by proposing \dsname, a frequency-aware ADS. 
The key idea is to decouple authentication from layout optimization, where cold elements remain in a conventional authenticated tier, frequently accessed elements move into a periodically rebuilt HuffMHT, and newly promoted hot elements are handled through an overflow tree until the next rebuild. 
This tiered design exploits access skew by giving hot elements shorter authenticated paths, while batching layout changes so that the tree does not need to be reshaped after every access.

Experiments on real-world data show that this design improves several metrics. 
The best-performing policy, Sliding-Window, uses about $2.4\times$ less average hash input than MPT ,and $0.34\times$ less than UBT. 
Its access-weighted proofs are also smaller by $0.18\times$ than MPT and $0.55\times$ than UBT.
The microbenchmarks support the same design choice at a smaller scale: batched pair-swap maintenance reduces layout restructuring time when the key set is stable, and periodic rebuilds avoid the high cost of updating after every operation under dynamic insertions.

\printbibliography[heading=bibintoc]

\clearpage

\appendix

\section{Algorithms}
\label[appendix]{sec:algorithms}

\begin{algorithm}
\captionDesc{Merkle Tree Inclusion Proof}\label{alg:mt-proof}
\SetAlgoLined
\SetAlgoVlined
\SetKwInput{KwInput}{Input}
\SetKwInput{KwOutput}{Output}
\KwInput{\\\quad
$\mathcal{T}$: Binary Merkle tree.
\\\quad
$c$: Trusted root commitment.
\\\quad
$v$: Claimed leaf value.
}
\KwOutput{A Merkle proof $\pi$, or an accept/reject bit.}

\BlankLine
\SetKwProg{Fn}{Procedure}{:}{end}

\Fn{\textsc{GetProof}($\mathcal{T},k$)}{
  $\pi \leftarrow [\,]$; \tcp{Initialize new list}
  $u \leftarrow$ leaf storing $(k,v)$\;
  \eIf{$(k,v) \in u$}{
    Append $(k,v)$ to $\pi$\;
  }{
    \Return $\bot$
  }
  \While{$u$ is not the root}{
    Append $(\mathsf{side}(u),\mathsf{hash}(\mathsf{sibling}(u)))$ to $\pi$\;
    $u \leftarrow$ parent of $u$\;
  }
  \Return $\pi$\;
}

\Fn{\textsc{Verify}($c,k, v,\pi$)}{
  $h \leftarrow H((k,v))$\;
  \ForEach{$(b,s)\in\pi$ from leaf to root}{
    \eIf{$b=\textsf{left}$}{
      $h \leftarrow H(h\,\|\,s)$\;
    }{
      $h \leftarrow H(s\,\|\,h)$\;
    }
  }
  \Return $h=c$\;
}
\end{algorithm}

\begin{algorithm}
\captionDesc{MPT Inclusion Proof}\label{alg:mpt-proof}
\SetAlgoLined
\SetAlgoVlined
\SetKwInput{KwInput}{Input}
\SetKwInput{KwOutput}{Output}
\KwInput{\\\quad
$\mathcal{T}$: Merkle Patricia Trie.
\\\quad
$c$: Trusted root hash.
\\\quad
$k$: Queried key.
\\\quad
$v$: Claimed value.
}
\KwOutput{An MPT proof $\pi$, or an accept/reject bit.}

\BlankLine
\SetKwProg{Fn}{Procedure}{:}{end}

\Fn{\textsc{GetProof}($\mathcal{T},k$)}{
  $p \leftarrow \textsc{Nibbles}(k)$\;
  $u \leftarrow$ root node of $\mathcal{T}$\;
  $\pi\leftarrow[\,]$; \tcp{Initialize new list}
  \While{$u$ is on the search path for $p$}{
    Append $\textsc{Encode}(u)$ to $\pi$\;
    $u \leftarrow$ next trie node selected by the remaining nibbles of $p$\;
  }
  \Return $\pi$\;
}

\Fn{\textsc{Verify}($c,k,v,\pi$)}{
  $p \leftarrow \textsc{Nibbles}(k)$\;
  $r \leftarrow c$\;
  \ForEach{encoded node $e\in\pi$ from root to leaf}{
    \If{$H(e)\ne r$}{\Return 0 \;}
    $u \leftarrow \textsc{Decode}(e)$\;
    \uIf{$u$ contains the value for the remaining path $p$}{
      $\hat v \leftarrow$ value stored at $u$\;
      \textbf{break}\;
    }
    \Else{
      $(r,p) \leftarrow$ next child commitment and remaining path selected by $p$\;
    }
  }
  \Return $\hat v=v$\;
}

\Fn{\textsc{Nibbles}($k=b_1\cdots b_m$)}{
\tcp{Each $b_i$ denotes one byte of the key.}
  \Return $\left(\left\lfloor \frac{b_1}{16}\right\rfloor, b_1\bmod 16,\dots,
  \left\lfloor \frac{b_m}{16}\right\rfloor, b_m\bmod 16\right)$\;
}
\end{algorithm}
\FloatBarrier

\begin{algorithm}
\captionDesc{\dsname}\label{alg:hmt-ops}
\SetAlgoLined
\SetAlgoVlined
\SetKwInput{KwState}{State}
\KwState{
\\\quad
$\mathcal{A}_i$: ADS storing tier $i$.
\\\quad
$P_i$: tier-$i$ Promotion Cache.
\\\quad
$r_i$: tier-$i$ root.
}

\BlankLine
\SetKwProg{Fn}{Procedure}{:}{end}

\Fn{\textsc{Put}($k,v$)}{
$i \leftarrow \textsc{Locate}(k)$\;
$\mathcal{A}_i.\textsc{Put}(k,v)$\;
\textsc{InsertInTier}$(i,k)$\;
}

\Fn{\textsc{Delete}($k$)}{
$i \leftarrow \textsc{Locate}(k)$\;
\If{$\mathcal{A}_i.\textsc{Get}(k)=\bot$}{\Return\;}
$\mathcal{A}_i.\textsc{Delete}(k)$\;
\textsc{Remove}$(k)$\;
}

\Fn{\textsc{Commit}()}{
\For{$i \leftarrow 0$ \KwTo $T-1$}{
  $r_i \leftarrow \mathcal{A}_i.\textsc{Commit}()$\;
}
$c \leftarrow H(r_0\,\|\,\cdots\,\|\,r_{T-1})$\;
\Return $c$\;
}

\Fn{\textsc{Get}($k$)}{
$i \leftarrow \textsc{Locate}(k)$\;
$v \leftarrow \mathcal{A}_i.\textsc{Get}(k)$\;
\If{$v=\bot$}{\Return $\bot$\;}
\eIf{$k$ is cached in $P_i$}{
  \textsc{TouchInTier}$(i,k)$\;
}{
  \textsc{InsertInTier}$(0,k)$\;
}
\Return $v$\;
}

\Fn{\textsc{GetProof}($k$)}{
$i \leftarrow \textsc{Locate}(k)$\;
$\pi_i \leftarrow \mathcal{A}_i.\textsc{GetProof}(k)$\;
\If{$\pi_i=\bot$}{\Return $\bot$\;}
$\Pi \leftarrow (i,\pi_i,(r_0,\dots,r_{T-1}))$\;
\Return $\Pi$\;
}

\Fn{\textsc{Verify}($c,k,v,\Pi$)}{
Parse $\Pi=(i,\pi_i,(r_0,\dots,r_{T-1}))$\;
\If{$\mathcal{A}_i.\textsc{Verify}(r_i,k,v,\pi_i)=0$ or $H(r_0\,\|\,\cdots\,\|\,r_{T-1}) \ne c$}{\Return $0$\;}
\Return $1$\;
}
\end{algorithm}

\begin{algorithm}
\captionDesc{Bucketed LFU Promotion Cache}\label{alg:lfu}
\SetAlgoLined
\SetAlgoVlined
\SetKwInput{KwConfig}{Configuration}
\SetKwInput{KwState}{State}
\KwConfig{\\\quad
$T$: Number of tiers.
\\\quad
$R$: Frequency-bucket span.
\\\quad
$C_i$: capacity of tier-$i$ metadata cache.
}
\KwState{
$\mathcal{S}$: Shared CMS.
\\\quad
$M_i$: Key-to-node hash map for tier $i$.
\\\quad
$B_i$: Doubly-linked list of tier $i$'s frequency buckets.
}

\BlankLine
\SetKwProg{Fn}{Procedure}{:}{end}

\Fn{\textsc{Locate}($k$)}{
\For{$i \leftarrow 0$ \KwTo $T-1$}{
  \If{$k \in M_i$}{
    \Return $i$\;
  }
}
\Return $0$ as base-tier fallback\;
}

\Fn{\textsc{Delete}($k$)}{
$i \leftarrow \textsc{Locate}(k)$\;
\If{$i=0 \land k\notin M_0$}{\Return not cached\;}
Delete $M_i[k]$ from $B_i$; delete $M_i[k]$\;
}

\Fn{\textsc{TouchInTier}($i,k$)}{
$n \leftarrow M_i[k]$\;
update $\hat f$ of $k$ in $\mathcal{S}$\;
\textsc{RelocateIfNeeded}$(i,n,\hat f)$\;
}

\Fn{\textsc{RelocateIfNeeded}($i,n,\hat f$)}{
\If{$\hat f$ is inside $n$'s current bucket range}{
  move $n$ to the front of that bucket list\;
}
\Else{
  Remove $n$ from current bucket\;
  $b \leftarrow$ find-or-create bucket in $B_i$ for range of $\hat f$ with span $R$\;
  Insert $n$ at front of $b$\;
}
}

\Fn{\textsc{InsertInTier}($i,k$)}{
\If{$k \in M_i$}{
  \textsc{TouchInTier}$(i,k)$\;
  \Return\;
}
Update $\hat f$ of $k$ in $\mathcal{S}$\;
$b \leftarrow$ find-or-create bucket in $B_i$ for range of $\hat f$ with span $R$\;
Create node $n=(k,b)$, insert at front of $b$\;
$M_i[k] \leftarrow n$\;
\If{$|M_i|>C_i$}{Evict tail of minimum-frequency bucket in $B_i$\;}
}

\Fn{\textsc{MoveTier}($k,i,j$)}{
Remove $k$ from $M_i$ and its bucket in $B_i$\;
\textsc{InsertInTier}$(j,k)$\;
}
\end{algorithm}

\begin{algorithm}
\captionDesc{Absolute-Threshold Migration}\label{alg:absolute-threshold}
\SetAlgoLined
\SetAlgoVlined
\SetKwInput{KwParam}{Parameters}
\SetKwInput{KwState}{State}
\KwParam{
\\\quad
$T$: Number of tiers.
\\\quad
$C_i$: Promotion Cache capacity for tier $i$.
\\\quad
$\theta_i$: Absolute promotion threshold for edge $i$.
}
\KwState{
\\\quad
$\mathcal{A}_i$: ADS storing tier $i$.
\\\quad
$P_i$: Promotion Cache for tier $i$.
}

\BlankLine
\tcp{Process each promotion edge}
\For{$i \leftarrow 0$ \KwTo $T-2$}{
  $x \leftarrow \textsc{MostFrequent}(P_i)$\;
  \tcp{1. Fill spare capacity in the hotter tier}
  \While{$x\ne\bot \land |P_{i+1}| < C_{i+1} \land \hat f(x)>\theta_i$}{
    Promote $x$ from $\mathcal{A}_i$ to $\mathcal{A}_{i+1}$\;
    \textsc{MoveTier}$(x,i,i+1)$\;
    $x \leftarrow \textsc{MostFrequent}(P_i)$\;
  }
  \tcp{2. If full, admit only if hotter than tail}
  \If{$|P_{i+1}| = C_{i+1}$}{
    $z \leftarrow \textsc{LeastFrequent}(P_{i+1})$\;
    \If{$x\ne\bot \land z\ne\bot \land \hat f(x)>\theta_i \land \hat f(x)>\hat f(z)$}{
      Promote $x$ from $\mathcal{A}_i$ to $\mathcal{A}_{i+1}$\;
      Demote $z$ from $\mathcal{A}_{i+1}$ to $\mathcal{A}_i$\;
      \textsc{MoveTier}$(x,i,i+1)$\;
      \textsc{MoveTier}$(z,i+1,i)$\;
    }
  }
}
\end{algorithm}

\begin{algorithm}
\captionDesc{Sliding-Window Migration}\label{alg:sliding-window}
\SetAlgoLined
\SetAlgoVlined
\SetKwInput{KwInput}{Input}
\SetKwInput{KwParam}{Parameters}
\SetKwInput{KwState}{State}
\KwInput{\\\quad
$B$: Access histogram for the completed period.
}
\KwParam{\\\quad
$T$: Number of tiers.
\\\quad
$\theta_i$: Window-rate threshold for edge $i$.
\\\quad
$W$: Sliding-window length, in periods.
\\\quad
$\Delta_i$: Delayed-demotion interval for edge $i$.
}
\KwState{\\\quad
$\mathcal{W}$: Ring of recent period histograms.
\\\quad
$H_{\mathcal{W}}$: Aggregate histogram over $\mathcal{W}$.
\\\quad
$\mathcal{Q}_i$: Demotion queue for edge $i$.
\\\quad
$\mathcal{A}_i$: ADS storing tier $i$.
}

\BlankLine
\tcp{1. Slide the window}
$E \leftarrow$ oldest period histogram in $\mathcal{W}$\;
$\mathcal{X}^{-} \leftarrow \{x\mid E[x]>0\}$\;
$\mathcal{X}^{+} \leftarrow \{x\mid B[x]>0\}$\;
$H_{\mathcal{W}} \leftarrow H_{\mathcal{W}}-E+B$\;
$\mathcal{W} \leftarrow \textsc{AppendNewest}(\mathcal{W}\setminus E,B)$\;
\BlankLine

\tcp{2. Recheck delayed demotions}
\For{$i \leftarrow T-2,T-3,\dots,0$}{
    \ForEach{expired candidate $x\in \mathcal{Q}_i$}{
      Dequeue $x$ from $\mathcal{Q}_i$\;
      \If{$\textsc{Locate}(x)=i+1 \land s_W(x)<\theta_i$}{
        Demote $x$ from $\mathcal{A}_{i+1}$ to $\mathcal{A}_i$\;
        \textsc{MoveTier}$(x,i+1,i)$\;
      }
    }
}
\BlankLine

\tcp{3. Schedule newly cold hot tier keys}
\ForEach{$x\in \mathcal{X}^{-}$}{
    $j \leftarrow \textsc{Locate}(x)$\;
    \If{$j>0$}{
      $i \leftarrow j-1$\;
      \If{$s_W(x)<\theta_i$}{
        \If{$x$ is not scheduled in $\mathcal{Q}_i$}{
          Enqueue $x$ in $\mathcal{Q}_i$ for time $+\Delta_i$\;
        }
      }
    }
}
\BlankLine

\tcp{4. Promote newly active qualifying keys}
\ForEach{$x\in \mathcal{X}^{+}$}{
    $i \leftarrow \textsc{Locate}(x)$\;
    \If{$i<T-1$}{
      \If{$s_W(x)\ge\theta_i$}{
        Promote $x$ from $\mathcal{A}_i$ to $\mathcal{A}_{i+1}$\;
        \textsc{MoveTier}$(x,i,i+1)$\;
      }
    }
}
\end{algorithm}

\FloatBarrier

\begin{algorithm}
\captionDesc{Dynamic-Control Controller}\label{alg:dynamic-control}
\SetAlgoLined
\SetAlgoVlined
\SetKwInput{KwInput}{Input}
\SetKwInput{KwParam}{Parameters}
\SetKwInput{KwState}{State}
\KwInput{\\\quad
$b$: Current period index.
}
\KwParam{\\\quad
$T$: Number of tiers.
\\\quad
$q_{\mathrm{high}}$: Rejection threshold.
\\\quad
$K$: Capacity-control interval.
\\\quad
$\theta_i^{\min}$: Lower bound for edge $i$.
\\\quad
$M$: Threshold-control interval.
\\\quad
$\theta_i^{\max}$: Tail-derived upper bound.
\\\quad
$U$: Low-occupancy streak length.
\\\quad
$\rho_{\mathrm{low}},\rho_{\mathrm{high}}$: Target occupancy band.
}
\KwState{\\\quad
$n_i$: Hot tier size.
\\\quad
$C_i$: Hot tier capacity.
\\\quad
$\theta_i$: Current threshold.
\\\quad
$A_i,R_i$: Attempts and rejects.
\\\quad
$u_i$: Low-occupancy streak.
}

\BlankLine
\tcp{Current control-window index}
\If{$b \bmod M \ne 0$}{
  \Return\;
}
$w_{\mathrm{ctrl}} \leftarrow \frac{b}{M}$\;
\BlankLine

\tcp{1. Threshold control}
\For{$i \leftarrow 0$ \KwTo $T-2$}{
  $\rho_i \leftarrow \frac{n_i}{C_i}$\;
  $q_i \leftarrow \frac{R_i}{\max(A_i,1)}$\;
  \If{$\rho_i>\rho_{\mathrm{high}} \land q_i>q_{\mathrm{high}}$}{
    Raise $\theta_i$, capped by $\theta_i^{\max}$\;
  }
  Update $u_i$: increment if $\rho_i<\rho_{\mathrm{low}}$, otherwise reset\;
  \If{$u_i\ge U$}{
    Lower $\theta_i$, floored by $\theta_i^{\min}$\;
  }
  Reset $A_i$ and $R_i$\;
}
\BlankLine

\If{$w_{\mathrm{ctrl}} \bmod K \ne 0$}{
  \Return\;
}
\BlankLine

\tcp{2. Capacity control}
\For{$i \leftarrow 0$ \KwTo $T-2$}{
  $\rho_i \leftarrow \frac{n_i}{C_i}$\;
  \If{$\rho_i>\rho_{\mathrm{high}} \land \theta_i$ is near $\theta_i^{\max}$}{
    Grow $C_i$\;
  }
  \If{$\rho_i<\rho_{\mathrm{low}} \land \theta_i$ is near $\theta_i^{\min}$}{
    Shrink $C_i$\;
  }
}
\end{algorithm}

\section{Proofs}
\label[appendix]{sec:proofs}
\resZipfAverageDepthRatio*
\begin{proof}
In a HuffMHT, an item's depth equals its Huffman code length.
For a Huffman code over distribution $P_N$ and given the base-2 entropy $H_2(P_N)=\sum_{i=1}^{N}p_i\log_2\frac{1}{p_i}$ the expected length is bounded by $\bar d_{\mathrm{Huff}}(N) < H_2(P_N)+1$.
Substituting $p_i=\frac{(i+b)^{-a}}{Z_N}$ gives $\frac{1}{p_i}=Z_N(i+b)^a$.
Hence,
\begin{align*}
H_2(P_N)
&=
\sum_{i=1}^{N}p_i
\log_2\!\left(Z_N(i+b)^a\right)\\
&=
\sum_{i=1}^{N}p_i
\left(\log_2 Z_N+a\log_2(i+b)\right)\\
&=
\log_2 Z_N\sum_{i=1}^{N}p_i
+
a\sum_{i=1}^{N}p_i\log_2(i+b)\\
&=
\log_2 Z_N
+
\frac{a}{Z_N}
\sum_{i=1}^{N}(i+b)^{-a}\log_2(i+b).
\end{align*}

We now show that the entropy is bounded independently of $N$.
As $a>1$, then $A_{a,b}\triangleq \sum_{i=1}^{\infty}(i+b)^{-a}<\infty$.
So, $Z_N=\sum_{i=1}^{N}(i+b)^{-a}\le A_{a,b}$ and also $Z_N\ge (1+b)^{-a}$, which, when taken together, imply $\frac{1}{Z_N}\le (1+b)^a$.
Next, since $a>1$, the logarithmic weighted tail is also summable:
$$
B_{a,b}
\triangleq
\sum_{i=1}^{\infty}(i+b)^{-a}\log_2(i+b)
<
\infty
.
$$
Thus, for every $N$,
$$
\sum_{i=1}^{N}(i+b)^{-a}\log_2(i+b)
\le
B_{a,b}
.
$$
By combining these bounds we obtain the following:
\begin{align*}
H_2(P_N)
&=
\log_2 Z_N
+
\frac{a}{Z_N}
\sum_{i=1}^{N}(i+b)^{-a}\log_2(i+b)
\\&
\le
\max\{0,\log_2 A_{a,b}\}
+
a(1+b)^a B_{a,b}
.
\end{align*}
Denote $C_{a,b} \define \max\{0,\log_2 A_{a,b}\}+a(1+b)^a B_{a,b}$, then $H_2(P_N)\le C_{a,b}$, and therefore:
$$
\bar d_{\mathrm{Huff}}(N)
<
H_2(P_N)+1
\le
C_{a,b}+1
=
O(1)
.
$$
Let $q\ge2$ be fixed.
A balanced $q$-ary MT with $N$ leaves has proof depth
$$
\bar d_{\mathrm{MT_q}}(N)
=
\lceil\log_q N\rceil
=
\Theta(\log N)
,
$$
where the hidden constants may depend on $q$ but not on $N$.
Thus,
$$
\frac{\bar d_{\mathrm{Huff}}(N)}
{\bar d_{\mathrm{MT_q}}(N)}
=
\frac{O(1)}{\Theta(\log N)}
=
O\!\left(\frac{1}{\log N}\right)
.
$$
If an MPT baseline satisfies $\bar d_{\mathrm{MPT}}(N)=\Omega(\log N)$ over the considered key-space family, then the same numerator bound gives
$$
\frac{\bar d_{\mathrm{Huff}}(N)}
{\bar d_{\mathrm{MPT}}(N)}
=
\frac{O(1)}{\Omega(\log N)}
=
O\!\left(\frac{1}{\log N}\right)
.
$$
\end{proof}

\resAdaptiveHuffmanCost*
\begin{proof}
Consider one adaptive Huffman update.
The update changes the accessed leaf weight, or splits the NYT leaf when inserting a new element, and may then perform swaps or slide-and-increment moves.
Let $A$ be the set of nodes whose authenticated leaf payload or ordered child relation changes during this process.
Only nodes on paths from nodes in $A$ to the root can receive different hash inputs.
Every node outside
\[
D=\bigcup_{v\in A}\mathsf{path}(v)
\]
has the same ordered children and the same descendant values before and after the update, so its hash remains valid.
Recomputing the nodes in $D$ once in bottom-up order therefore restores the authenticated root, and costs $O(|D|)$ hash computations.
If the update performs $s$ local structural moves, then the initial leaf update or NYT split contributes at most one affected path, and each move contributes only a constant number of additional affected paths.
Each such path has length at most $h$, so $|D|\le O((s+1)h)$.
On the other hand, $D$ is a subset of the nodes of the tree.
A binary tree with $n$ leaves has at most $2n-1$ nodes, and therefore $|D|\le O(n)$.
Combining these two bounds gives $|D|\le O(\min\{(s+1)h,n\})$.
The worst-case authenticated maintenance cost is therefore $O(n)$.
This is tight up to constants because an update or insertion at a leaf of depth $h$ dirties the full leaf-to-root path, and an adaptive Huffman tree can have height $h=n-1$.
\end{proof}

\section{Adaptive Huffman Replay Stress Test}
\label[appendix]{sec:adaptive-huffman-replay-stress}

We include a short real-trace stress test against adaptive Huffman to isolate the cost of online Huffman-tree reshaping.
The replay starts from the same Ethereum block as \cref{sec:ethereum-replay-evaluation} and processes the next 10{,}000 blocks.
We run the same account-level operation stream through over a balanced MT, a periodic-rebuild HuffMHT that rebuilds once every 500 blocks, and an adaptive Huffman tree that updates its layout after each access or insertion.

\Cref{fig:adaptive_huffman_stress_throughput} shows that the per-operation tree maintenance in adaptive Huffman dominates its replay cost.
The balanced MT processes 1.14M operations per second, and periodic-rebuild HuffMHT processes 667K operations per second.
By contrast, Adaptive Huffman processes only 2.3K operations per second, despite producing the smallest weighted proofs in \cref{fig:adaptive_huffman_stress_proof_size}.
Periodic-rebuild HuffMHT therefore captures much of the proof-size benefit of a Huffman-shaped layout, while avoiding the cost of reshaping and rehashing the authenticated tree after every operation.

\begin{figure}
    \centering
    \begin{minipage}{0.48\columnwidth}
        \centering
    \includegraphics[width=\linewidth]{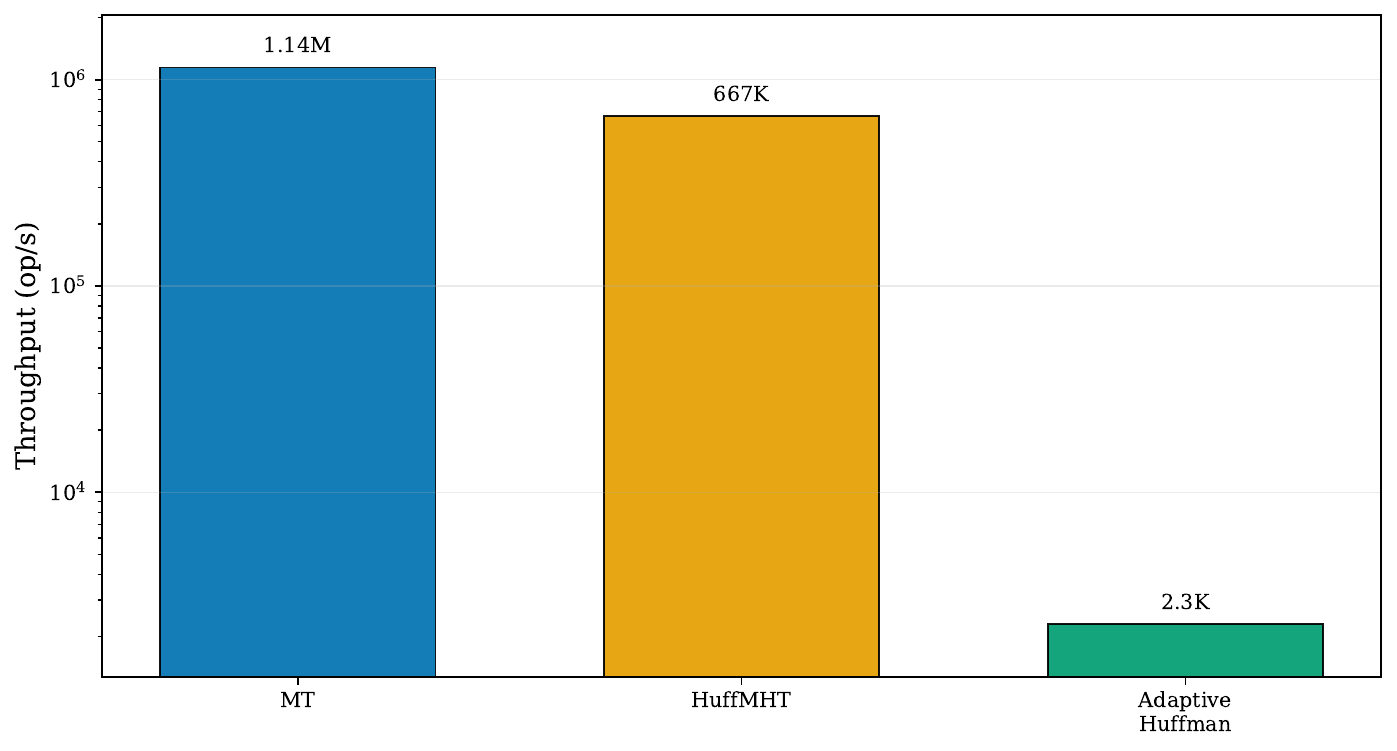}
    \captionDesc{Throughput for MT, periodic-rebuild HuffMHT, and Adaptive Huffman on the 10{,}000-block Ethereum stress test.}
    \label{fig:adaptive_huffman_stress_throughput}
    \end{minipage}
    \hfill
    \begin{minipage}{0.48\columnwidth}
        \centering
    \includegraphics[width=\linewidth]{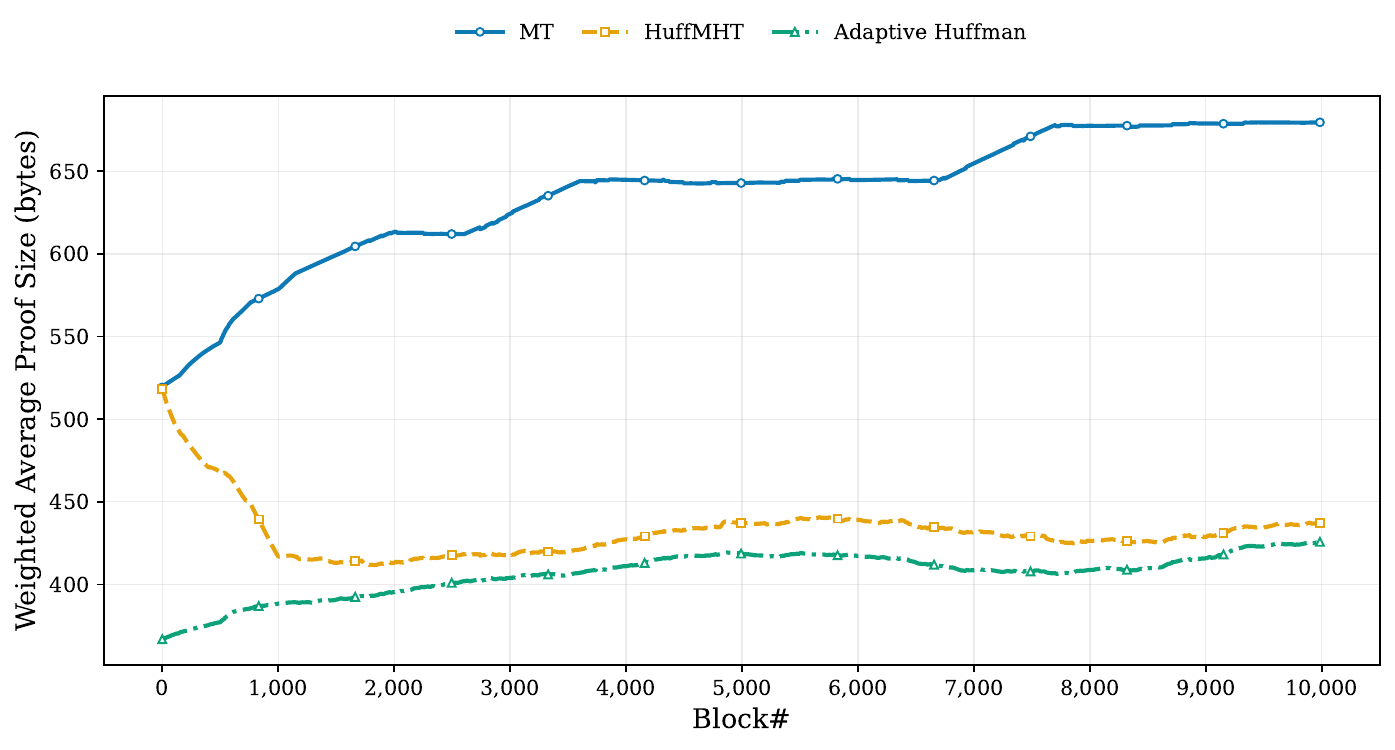}
    \captionDesc{Access-weighted average proof size for MT, periodic-rebuild HuffMHT, and Adaptive Huffman on the 10{,}000-block Ethereum stress test.}
    \label{fig:adaptive_huffman_stress_proof_size}
    \end{minipage}
\end{figure}

\section{Tier-Count Sensitivity}
\label[appendix]{sec:tier-count-sensitivity}

We evaluate additional \dsname configurations with more than two tiers to validate the design choice of two-tier.
For each policy, the multi-tier configuration extends the corresponding two-tier configuration used in \cref{sec:ethereum-replay-evaluation} or \cref{sec:additional-policy-replay}.
All additional hotter tiers are implemented as HuffMHTs.
The Promotion Cache capacities remain $8{,}000$ for the cold tier and $16{,}000$ for each hotter tier, and the 500-block batch/rebuild cadence is unchanged.
The added tiers only introduce additional promotion boundaries.
For these boundaries, we start from the policy's two-tier threshold and double it at each hotter boundary; for example, a six-tier rate-threshold policy uses thresholds $0.05,0.10,0.20,0.40,0.80$, while Absolute-Threshold uses $50{,}000,100{,}000,200{,}000,400{,}000,800{,}000$.
For each block, the replayer first computes the access-weighted average inclusion-proof size over the keys accessed in that block.
The y-axis in \cref{fig:tier_count_proof_size} then reports the mean of these per-block weighted averages over the full replay.

\begin{figure}
    \centering
    \includegraphics[width=0.99\columnwidth]{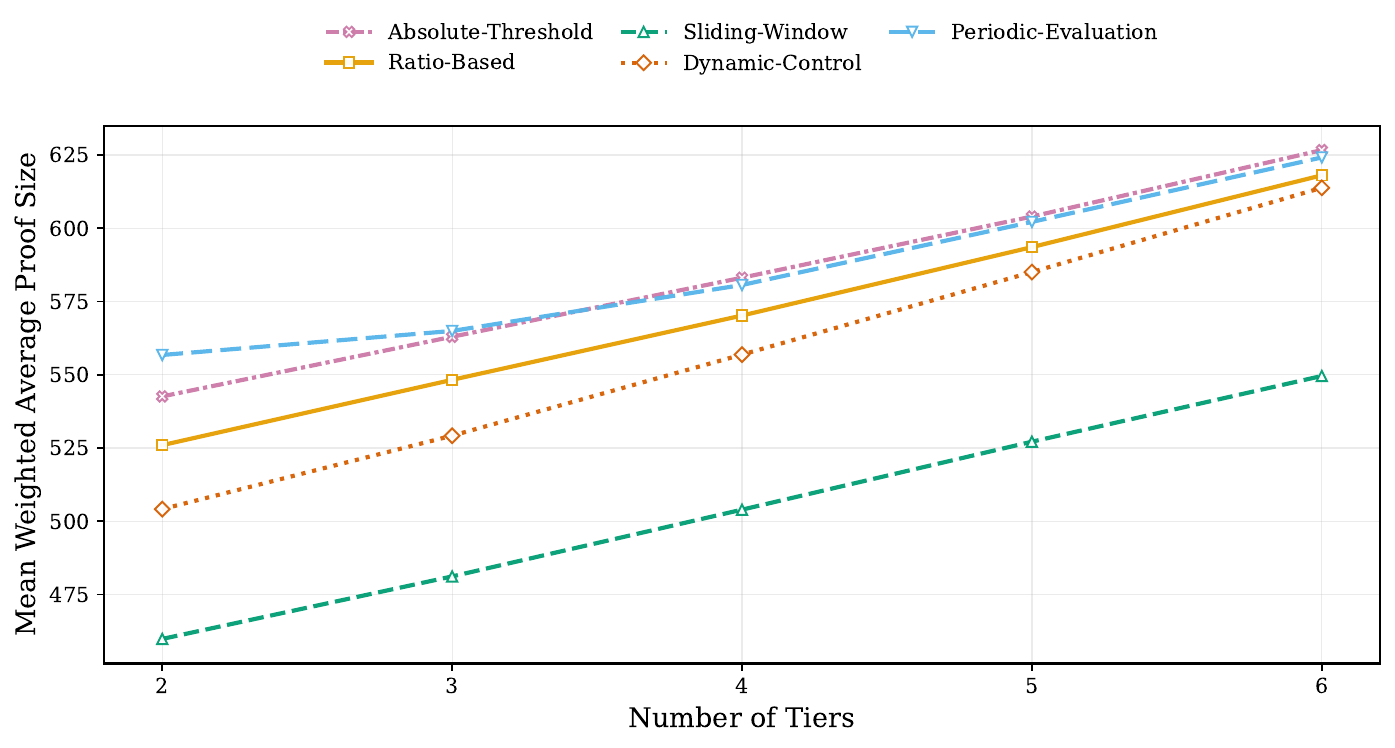}
    \captionDesc{Tier-count sensitivity for \dsname policies in the one-million-block Ethereum replay. Each point reports the full-replay mean of the per-block access-weighted average proof size.}
    \label{fig:tier_count_proof_size}
\end{figure}

\Cref{fig:tier_count_proof_size} shows that increasing tier number does not improve mean proof size.
Although additional tiers can shorten the tier-local proof for the hottest elements, each proof must also include the root of each tier to reconstruct the global root.
The contribution of this cross-tier proof component outweighs the marginal local-proof savings from further splitting the hot set.

\section{Evaluation of Additional Policies}
\label[appendix]{sec:additional-policy-replay}
We now evaluate two additional fixed-policy configurations using the same one-million-block Ethereum data set: Absolute-Threshold and Periodic-Evaluation.
This is in addition to \cref{sec:ethereum-replay-evaluation} which focused on three \dsname policies: Ratio-Based, Sliding-Window, and Dynamic-Control.
The two policies use the two-tier \dsname design with a balanced MT for the cold tier, a periodic-rebuild HuffMHT as hot tier, fixed cache capacities of $8{,}000/16{,}000$ entries, a cache frequency range of 10, and 500-block batches.
Absolute-Threshold promotes an element when its cumulative CMS estimate exceeds $50{,}000$.
Periodic-Evaluation uses the interval-based policy, evaluates promotion candidates every 500 blocks, and promotes an element upon meeting a rate threshold of $0.05$.
\Cref{fig:policy_replay_hash,fig:policy_replay_proof} show that these additional fixed policies preserve the main evaluation trend.
Both Absolute-Threshold and Periodic-Evaluation remain below MPT and UBT in hashed input bytes, while their access-weighted proof sizes stay clustered with the other \dsname variants.

\begin{figure}
    \centering
    \begin{minipage}{0.48\columnwidth}
\centering
    \includegraphics[width=\linewidth]{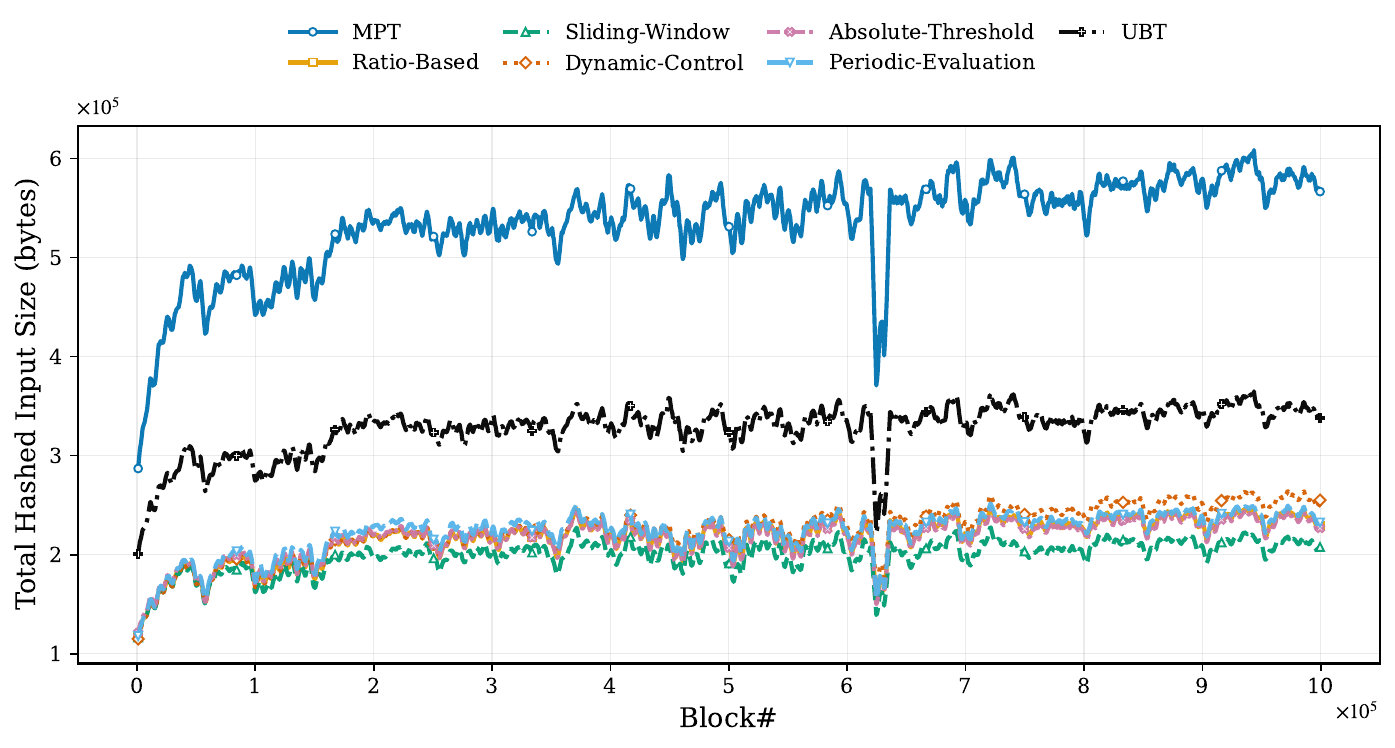}
    \captionDesc{Per-block hashed inputs comparison on one million Ethereum blocks.}
    \label{fig:policy_replay_hash}
    \end{minipage}
    \hfill
    \begin{minipage}{0.48\columnwidth}
    \centering
    \includegraphics[width=\linewidth]{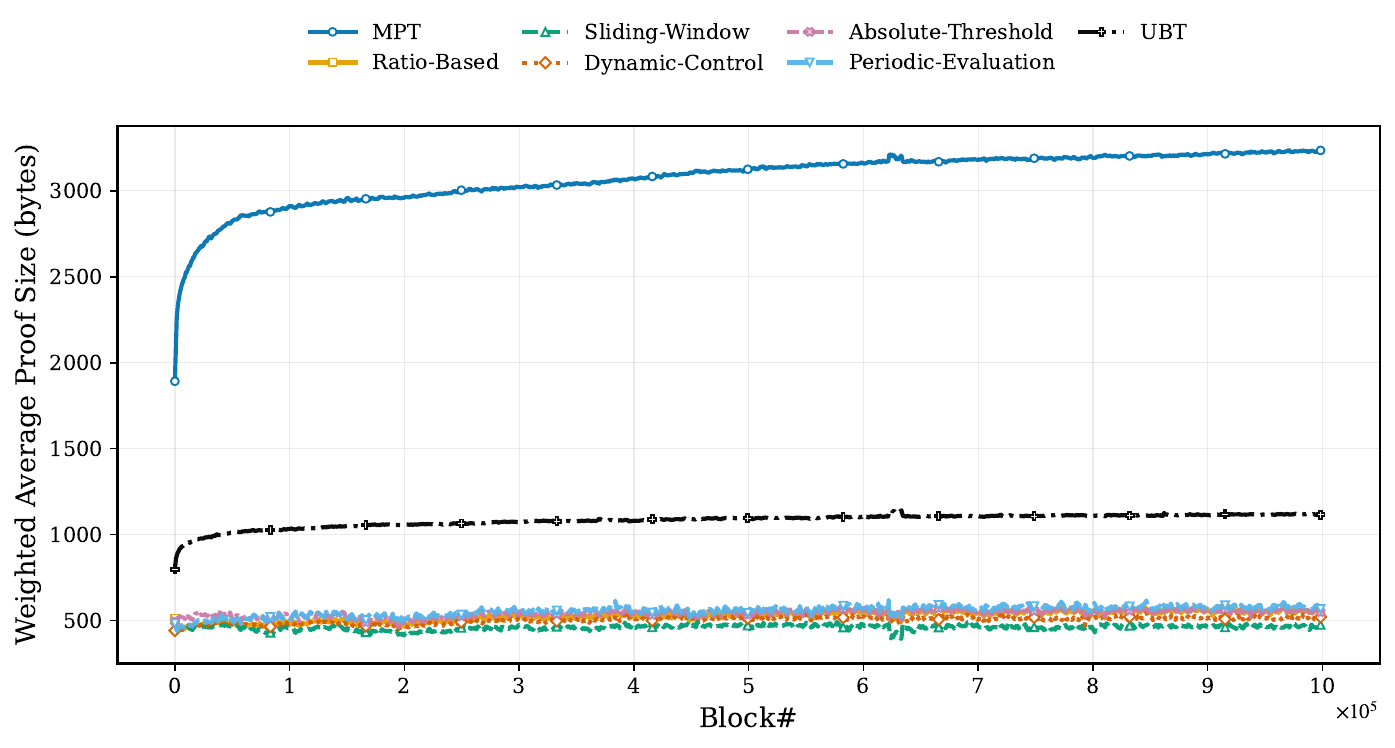}
    \captionDesc{Access-weighted average proof size comparison on one million Ethereum blocks.}
    \label{fig:policy_replay_proof}
    \end{minipage}
\end{figure}

\section{Underlying ADSs Choice}
\label[appendix]{sec:underlying-ads-choice}

We also evaluate the two-tier tree combination that motivates using a HuffMHT for hot elements and an MT for cold and newly inserted elements.
The replay uses the Absolute-Threshold configuration from \cref{sec:additional-policy-replay}, with the same Ethereum starting block as \cref{sec:ethereum-replay-evaluation}.
We compare three two-tier combinations, written as hot tier+cold tier: MT+MT, HuffMHT+MT, and HuffMHT+HuffMHT.

\begin{figure}
    \centering
    \begin{minipage}{0.48\columnwidth}
    \centering
    \includegraphics[width=0.99\columnwidth]{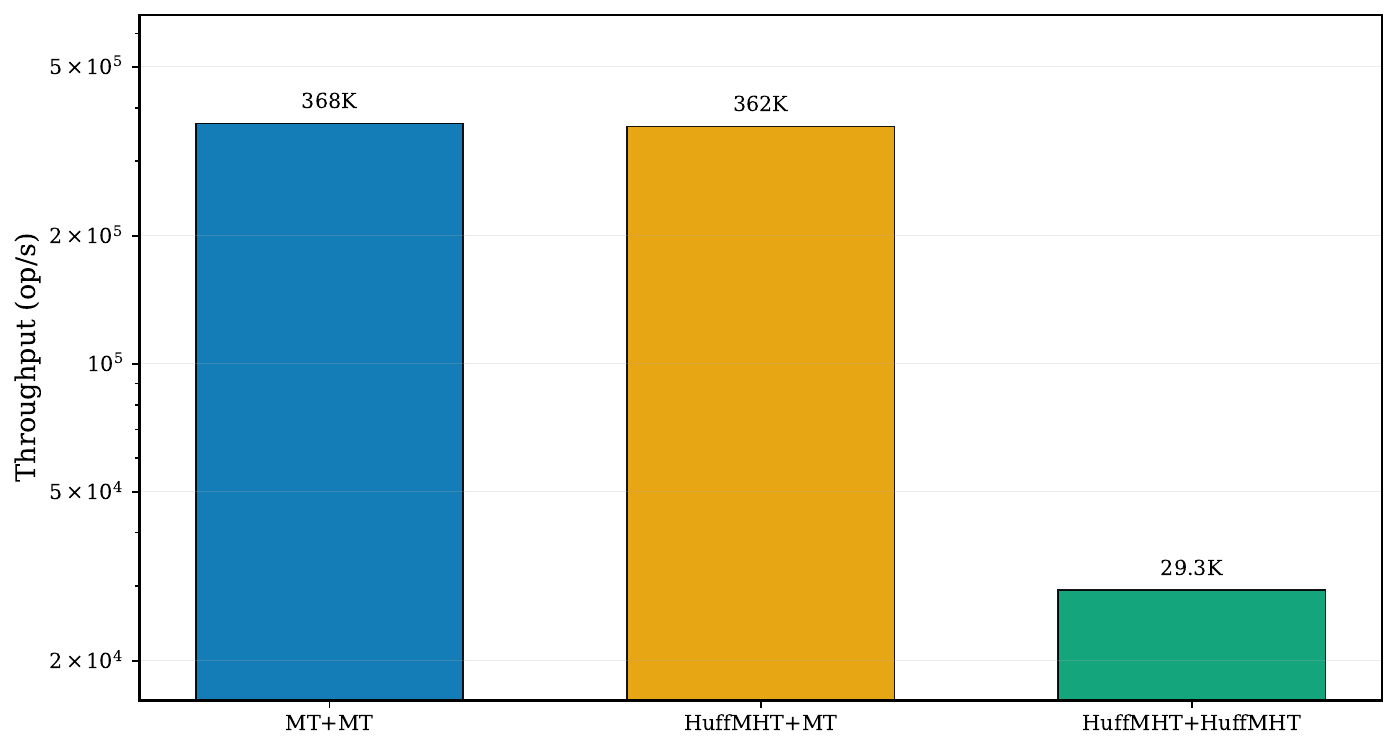}
    \captionDesc{Throughput for tier combinations under the Absolute-Threshold configuration. 
    Each label defines the tiers, e.g., HuffMHT+MT means that the hot tier uses HuffMHT and the cold one uses MT.}
    \label{fig:tree_combination_throughput}
    \end{minipage}
    \hfill
    \begin{minipage}{0.48\columnwidth}
    \centering
    \includegraphics[width=0.99\columnwidth]{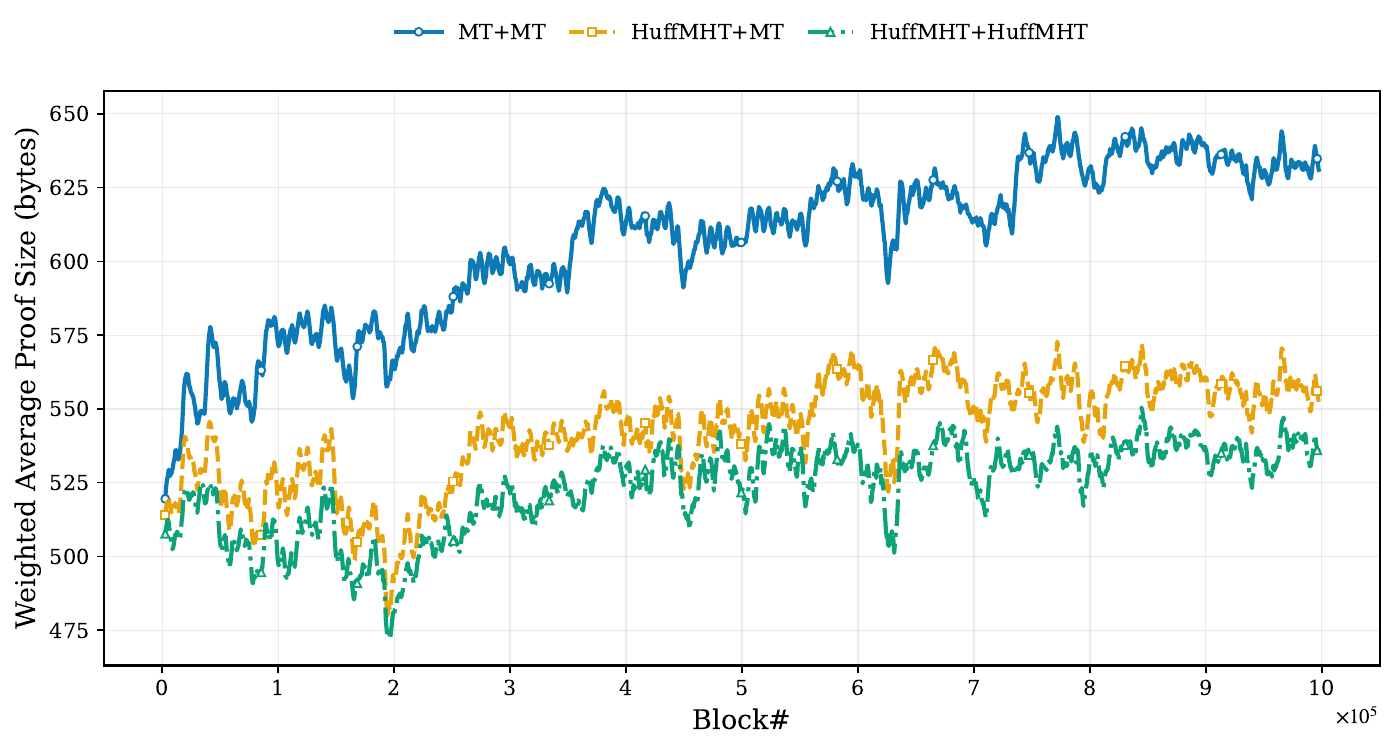}
    \captionDesc{Access-weighted average proof size for tier combinations under the Absolute-Threshold configuration.
Labels denote hot-tier+cold-tier choices; HuffMHT+MT is the \dsname combination.}
    \label{fig:tree_combination_proof_size}
    \end{minipage}
\end{figure}

\Cref{fig:tree_combination_throughput} shows that MT+MT and HuffMHT+MT have comparable throughput, at 368K and 362K operations per second, respectively.
HuffMHT+HuffMHT reaches only 29.3K operations per second, about $12.3\times$ lower than HuffMHT+MT.
\Cref{fig:tree_combination_proof_size} shows the complementary proof-size tradeoff.
MT+MT has the largest weighted proofs because neither tier shortens paths for frequently accessed elements.
HuffMHT+HuffMHT gives the smallest weighted proofs, while HuffMHT+MT remains close to it and below MT+MT.
Thus, MT+MT is efficient but does not reduce proof sizes for hot elements, whereas making both tiers Huffman-shaped substantially reduces throughput.

\section{Weighted Average Proof Size By Read/Write}
\label[appendix]{sec:read-write-proof-size}

The proof-size metric used in \cref{sec:ethereum-replay-evaluation} weights each key by the total number of both write and read accesses in the block.
We now consider refinements of this metric which account for only one of the two: either reads or writes.
We evaluate each refinement and present the results in \cref{fig:proof_weighted_by_read,fig:proof_weighted_by_write}, noting that both show the same performance ordering as the aggregate metric.
MPT has the largest proofs across metrics, rising from roughly 2 kilobytes at the beginning of the replay to above 3 kilobytes later in the run.
UBT lowers the read- and write-weighted averages to about 1 kilobyte, but remains well above the \dsname variants.
All \dsname policies remain in the hundreds-of-bytes range for both reads and writes, with Sliding-Window generally at or near the lowest curve.
The similarity between the two metrics indicates that the proof-size reduction is not primarily a consequence of the trace's read/write composition.
Rather, the reduction persists for both operations because frequently accessed accounts are placed on shorter paths in the hot tier.

\begin{figure}
    \centering
    \begin{minipage}{0.48\columnwidth}
    \centering
    \includegraphics[width=\linewidth]{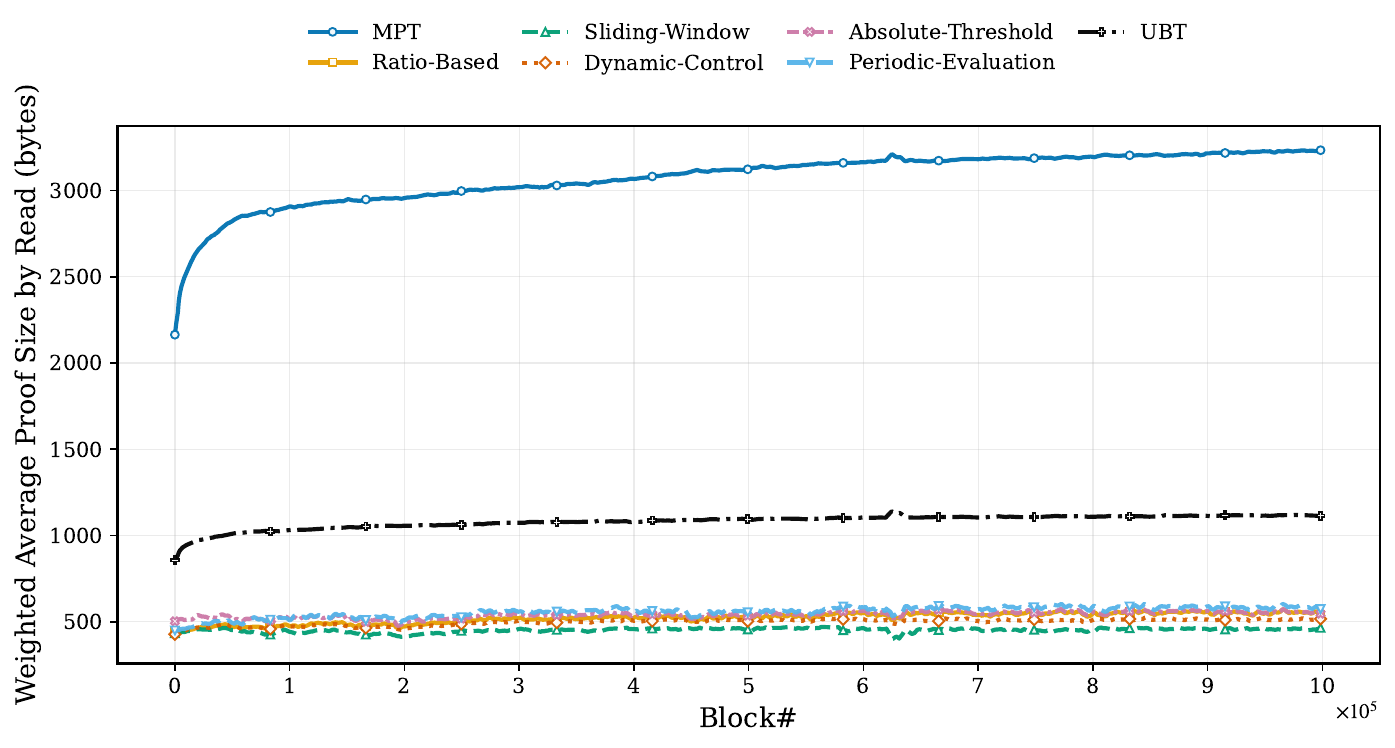}
    \captionDesc{Read-weighted average proof size in the one-million-block Ethereum replay.}
    \label{fig:proof_weighted_by_read}
    \end{minipage}
    \hfill
    \begin{minipage}{0.48\columnwidth}
    \centering
    \includegraphics[width=\linewidth]{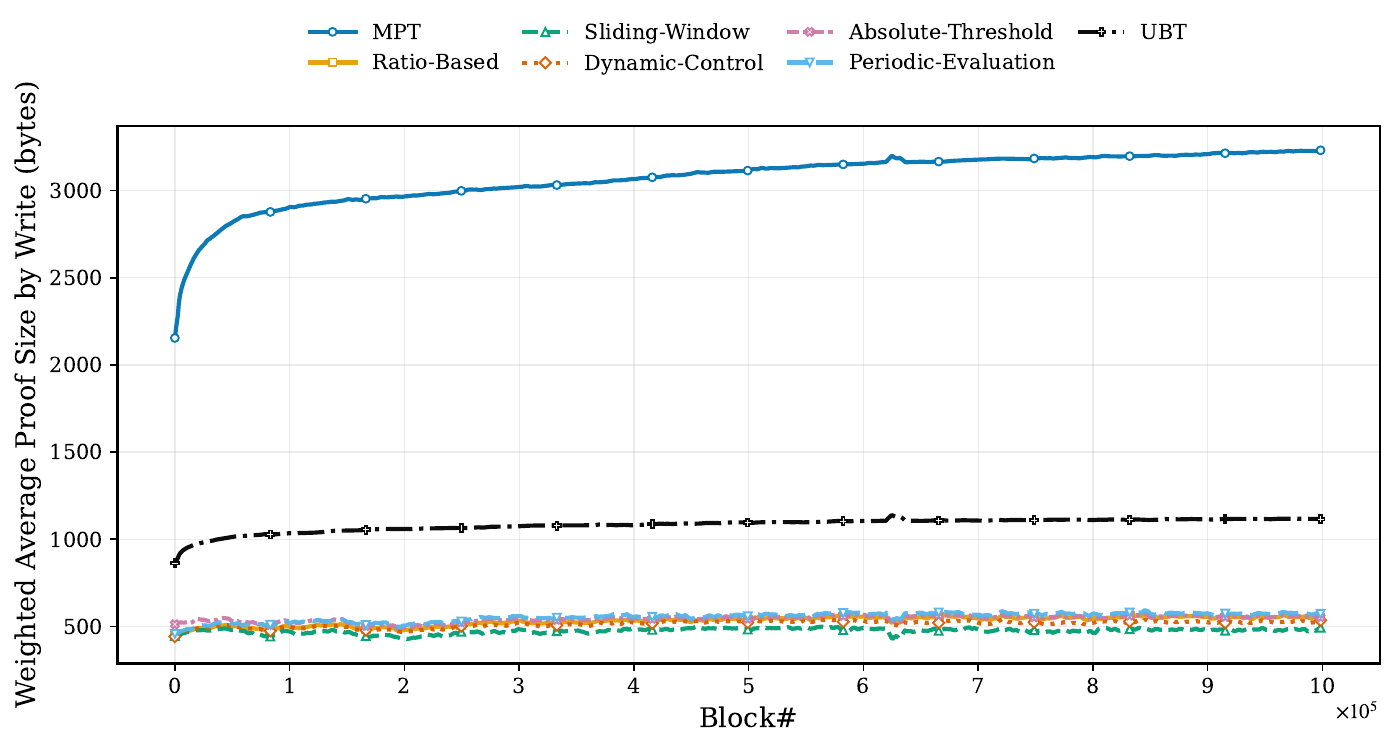}
    \captionDesc{Write-weighted average proof size in the one-million-block Ethereum replay.}
    \label{fig:proof_weighted_by_write}
    \end{minipage}
\end{figure}
\end{document}